\documentclass{amsart}[12pt]
\usepackage{amsmath}
\usepackage{graphicx}
\usepackage[titletoc,toc]{appendix}
\usepackage{latexsym}
\usepackage{hyperref}
\usepackage{graphicx}
\usepackage{amsthm,amssymb,amstext,amscd,amsfonts,amsbsy,amsxtra,latexsym}
\usepackage{wrapfig}
\usepackage{color}
\usepackage{pdflscape}
\usepackage{lscape}
\usepackage{tabularx}
\usepackage{lipsum}
\usepackage{rotating} 
\usepackage{pifont}
\usepackage{upgreek}
\usepackage{comment}
\usepackage{mathrsfs, mathtools}
\usepackage{wasysym}
\usepackage[top=2.5cm, bottom=2.5cm, left=2.5cm, right=2.5cm]{geometry}
\usepackage[normalem]{ulem}

\newcommand{\lb}{\label}
\newcommand{\bea}{\begin{eqnarray}}
\newcommand{\eea}{\end{eqnarray}}

\newcommand{\lp}[2]{\left \Vert \, #1 \, \right \Vert_{#2}}

\newcommand{\chiint}{\chi_{\text{int}}}
\newcommand{\chiext}{\chi_{\text{ext}}}

\makeatletter

\renewcommand{\tocsubsection}[3]{%
  \indentlabel{\@ifnotempty{#2}{\hspace*{2.3em}\makebox[2.3em][l]{%
    \ignorespaces#1 #2.\hfill}}}#3}
\renewcommand{\tocsubsubsection}[3]{%
  \indentlabel{\@ifnotempty{#2}{\hspace*{4.6em}\makebox[3em][l]{%
    \ignorespaces#1 #2.\hfill}}}#3}
		
\newcounter{mnotecount}[section]

\renewcommand{\themnotecount}{\thesection.\arabic{mnotecount}}

\newcommand{\mnote}[1]
{\protect{\stepcounter{mnotecount}}$^{\mbox{\footnotesize $%
\!\!\!\!\!\!\,\bullet$\themnotecount}}$ \marginpar{
\raggedright\tiny\em $\!\!\!\!\!\!\,\bullet$\themnotecount: #1} }

\usepackage[table]{xcolor}
\usepackage[normalem]{ulem}

\usepackage{tabulary} 
\usepackage{booktabs}

\theoremstyle{plain}
\newtheorem{theorem}{Theorem}[section]
\newtheorem{proposition}[theorem]{Proposition}
\newtheorem{lemma}[theorem]{Lemma}
\newtheorem{corollary}[theorem]{Corollary}

\theoremstyle{definition}

\newtheorem{remark}[theorem]{Remark}

\numberwithin{equation}{section}

\allowdisplaybreaks[3]

\swapnumbers
\DeclarePairedDelimiter\autobracket{(}{)}

\newcommand{\br}[1]{\autobracket*{#1}}

\newcommand{\metric}{{\mathbf g}}

\begin{document}

\begin{center}

\title[]{Global existence for a Fritz John equation in expanding FLRW spacetimes}

\author{}

\address{}

\email{}

\date{}

\parskip = 0 pt

\maketitle

\bigskip
João L.~ Costa \footnote{e-mail address: jlca@iscte-iul.pt}{${}^{,\sharp,\star}$},  
Jes\'{u}s Oliver \footnote{e-mail address: jesus.oliver@csueastbay.edu}{${}^{,\dagger}$} and 
Flavio Rossetti \footnote{{e-mail address: flavio.rossetti@gssi.it}}{${}^{, \S, \star}$}

\bigskip
{\it {${}^\sharp$}Departamento de Matemática,}\\
{\it Instituto Universit\'ario de Lisboa (ISCTE-IUL),}\\
{\it Av. das For\c{c}as Armadas, 1649-026 Lisboa, Portugal}

\bigskip
{\it {${}^\star$}Center for Mathematical Analysis, Geometry and Dynamical Systems,}\\
{\it Instituto Superior T\'ecnico, Universidade de Lisboa,}\\
{\it Av. Rovisco Pais, 1049-001 Lisboa, Portugal}

\bigskip
{\it {${}^\dagger$} California State University East Bay,}\\
{\it 25800 Carlos Bee Boulevard,}\\
{\it Hayward, California, USA, 94542}

\bigskip
{\it {${}^\S$} Gran Sasso Science Institute,}\\
{\it Viale Francesco Crispi 7,}\\
{\it 67100 L'Aquila (AQ), Italy}

\end{center}\begin{abstract}
We study the family of semilinear wave 
 equations
\[\square_{\mathbf{g}_p}\phi=(\partial_t\phi)^2 \ , \]
on fixed expanding FLRW spacetimes, having $\mathbb{R}^3$ spatial slices and undergoing a power law expansion, with scale factor $a(t)=t^p$,
$0<  p \le 1$. 
This is a natural generalization to a non-stationary background of a famous Fritz John ``blow-up'' equation in $\mathbb{R}^{1+3}$ (corresponding to $p=0$, i.e.\ the case in which  $\mathbf{g}_0$ is the Minkowski metric). 

While, in Minkowski spacetime ($p=0$), non-trivial solutions to this equation are known to diverge in finite time, here  we prove that, on the referred FLRW backgrounds ($0<p\leq 1$),    
sufficiently small, smooth, and compactly supported initial data 
yield global-in-time solutions to the future. 
%
Previous work, co-authored by the first two authors \cite{jaj}, considered accelerated expanding spacetimes ($p>1$) and relied on the integrability of the inverse of the scale factor to establish future global well-posedness. 
In the current work, where such an integrability condition is lacking, we rely on a vector field method that captures and combines dispersive estimates with the spacetime expansion to control the solution and suppress the nonlinear blow-up mechanism. To achieve this, we commute the Laplace-Beltrami operator with a boosts-free subset of the Poincaré algebra and employ Klainerman-Sideris types of inequalities. 

Our strategy is general and is developed to handle the non-stationary nature of  FLRW  spacetimes. 
While we focus solely on this Fritz John type of equation, which serves as a prototype to study blow-up of non-linear waves, our approach provides a rigorous proof of the regularizing effects of spacetime expansion and can be exploited for a wider range of applications and nonlinearities. 

A consequence of our result is small data global existence for a class of wave equations with scale-invariant, time-dependent damping terms and quadratic time-derivative non-linearity in $\mathbb{R}^{1+3}$. This case lies in a range left open in the literature on the Glassey exponent for scale-invariant damped wave equations.
\end{abstract}


%
\section{Introduction}
%
We consider a framework in which a wavelike singularity formation mechanism is put in competition with the regularization effects created by cosmological expansion. As a blow up mechanism we consider the one provided by the famous Fritz John equation 
\begin{equation}
\label{JohnEq0}
\square_{\eta} \phi = (\partial_t\phi)^2\;,
\end{equation}
where $\eta$ is the $(1+3)$-dimensional Minkowski metric, and $\square_{\eta}  = -\partial^2_t + \Delta$ is the Laplace-Beltrami operator of the (Lorentzian) metric $\eta$~\footnote{This might seem like a pedantic way of introducing the classical wave operator, but it is just a convenient way to prepare the generalization of this equation to other geometries, by using the corresponding Laplace-Beltrami operator.}. A celebrated classical result~\cite{john_blow} shows, in particular, that all nontrivial solutions of this equation, with smooth and compactly supported initial data, blow up in finite time. To study the regularization effect of cosmological expansion, we generalize the wave dynamics determined by~\eqref{JohnEq0} into the setting of Friedmann-Lema\^itre-Robertson-Walker (FLRW) spacetimes, which, as we will discuss in a moment, is a family of physically motivated geometries describing expanding universes, parameterized by $p>0$, a parameter that controls the expansion rate. More precisely, we consider spacetime metrics of the form  
\begin{equation} \label{metric_intro0}
   \metric = \metric_p  = -dt^2 + t^{2p} \br{(dx^1)^2 + (dx^2)^2 + (dx^3)^2},
\end{equation}
where  $(t, x) = (t, x^1, x^2, x^3) \in\, \mathbb{R}^{+}\times \mathbb{R}^3$. Note that $\metric_0 = \eta$ is the Minkowski metric. We can then consider the Fritz John  equation as the $p=0$ element of the one-parameter family of partial differential equations (PDEs)   
\begin{equation}
\label{fritz_john_intro}
\square_{\metric_p} \phi = (\partial_t\phi)^2\;. 
\end{equation}
As already mentioned, we use the Laplace-Beltrami operator $\square_{\metric_p}$ as a geometric way to connect a wave operator in FLRW to the classical wave operator in Minkowski; this motivates the left-hand side of~\eqref{fritz_john_intro}. The choice of nonlinearity, in the right-hand side, is mainly motivated by the fact that the vector field $\partial_t$ is, for all $p\geq0$, the unit normal to the ``physical'' time $t=\text{constant}$ slices~\footnote{Moreover, this choice of nonlinearity $(\partial_{t}\phi)^2$ is a particular case of the class of nonlinearities considered in  \cite{jaj}, which concerns the case of accelerated expansion, $p>1$. Notwithstanding, other choices of nonlinearities, in the FLRW setting, also connect ``naturally'' to the original Fritz John case: for instance  $(\partial_{\tau}\phi)^2=t^{2p}(\partial_{t}\phi)^2$, where $\tau$ is the conformal time defined in~\eqref{def_tau}.}. For these reasons, we choose this family of PDEs as a ``natural'' generalization of~\eqref{JohnEq0} to the FLRW setting.

This is a class of geometric (non-linear) wave equations on curved backgrounds. By using the forthcoming formula~\eqref{flrw_eq}, in~\eqref{fritz_john_intro}, these PDEs can be cast as non-linear damped wave equations, in $\mathbb{R}^{+}\times \mathbb{R}^3$, with time-dependent coefficients
\begin{equation} \label{damped_eqn}
-\partial^2_t \phi-\frac{3 p}{t}\partial_t \phi+ \frac{1}{t^{2p}}\Delta \phi = (\partial_t\phi)^2\;.
\end{equation}
Even though, in such a formulation the relativistic setting could, in principle, be neglected, framing it in the context of general relativity provides physical and geometric interpretations of some key features of the dynamics. With that in mind, we briefly outline the cosmological context in which our geometric setting fits in:

According to the Cosmological Principle~\cite{Mukhanov2005}, an assumption with philosophical origins that is, by now, backed by strong empirical evidence, the Universe is,  on large scales, spatially homogeneous and isotropic. These symmetry assumptions strongly restrict the topology and geometry of cosmological models; for instance, if, from the few consistent choices, we pick a model with spatial $\mathbb{R}^3$ topology, then the corresponding spacetime Lorentzian metric can be written in the remarkably simple form
\begin{equation} \label{metric_intro}
   \mathbf{g}_{\text{FLRW}}  = -dt^2 + a(t)^2 \br{(dx^1)^2 + (dx^2)^2 + (dx^3)^2},
\end{equation}
with $(t, x) = (t, x^1, x^2, x^3) \in\, \mathbb{R}^{+}\times \mathbb{R}^3$, and a single free function: the scale factor $t\mapsto a(t)$. 
To obtain a complete cosmological model, one also needs to choose an appropriate matter model, for instance, a perfect fluid with linear equation of state, and then insert that information, together with the metric~\eqref{metric_intro}, into the Einstein equations. These procedure leads to an ODE system, the Friedmann equations, see e.g.\ \cite{FriedmannAllDimensions}, in terms of the scale factor, the pressure and density of the fluid.
%
%
%

In this work we focus on the following class of scale factors arising from the Friedmann equations:
\begin{equation} \label{scale_factor_intro}
a(t) = t^p, \quad 0 < p \le 1\;.
\end{equation}
When this factor is plugged in the FLRW metric~\eqref{metric_intro} we recover~\eqref{metric_intro0}, that describes a spacetime undergoing a polynomial, \textbf{non-accelerated} expansion of non-compact spatial sections, with $\mathbb{R}^3$ topology. This regime is physically relevant: the case $p=1/2$ corresponds to a radiation-filled spacetime, whereas $p=2/3$ is identified with a universe containing a dust (i.e.\ pressureless) fluid. These two settings  play a crucial role in the very successful $\Lambda$CDM model in cosmology \cite{Mukhanov2005}.

We should make it clear that the structure and dynamics of the fluid matter content of FLRW cosmological solutions will play no role to us here; they were just presented to provide a physical motivation for our use of FLRW geometries. 

As already mentioned, here we are interested in the behaviour of non-linear waves on these spacetimes, as dictated by~\eqref{fritz_john_intro}. 

For this purpose, it is useful to compare the current setting with some of the known results for waves propagating in such backgrounds~\cite{NatarioRossetti2023}. As we increase $p$, linear solutions decay ``slowly'', in time, for large expansion rates and, in the accelerated regime, $p > 1$, no decay of the $L^{\infty}$ norms occurs at all (See Section~\ref{secTechniques} for more details). Nonetheless, the decay of derivatives increases with $p$ due to the expansion. This leads to an increased decay in energy that facilitates the control of solutions and allows to prove global existence results for certain wave equations with quadratic non-linearities depending on derivatives, at least for $p > 1$ ~\cite{jaj}.
The situation concerning non-accelerated expansion, $0<p\leq 1$, and quadratic non-linearities remains, to the best of our knowledge, mostly unclear due to the slower decay in energy.

This is where our work comes in: focusing on the setting of~\eqref{fritz_john_intro}, it turns out, somewhat surprisingly, that any amount of expansion, i.e. any $p>0$, has a remarkable regularization effect. In fact, the main result of this work shows that, in contrast with the blow-up instabilities occurring for the semilinear wave equation~\eqref{JohnEq0} in the (static) $(1+3)$-dimensional Minkowski case, any amount of spacetime expansion leads to global existence of solutions, even in our non-accelerated setting~\footnote{It is interesting to compare with what happens for solutions of~\eqref{JohnEq0} with the increase of spatial dimensions $n$: in an appropriate regularity class, for $n=3$, we have the famous blow up result concerning all non-trivial solutions~\cite{john_blow}, while, for $n=4$, all sufficiently small data solutions are global in time~\cite[Theorem 2.1]{sogge}. Of course, this increase in dimension is discrete, while the expansion parameter $p$ under consideration in continuous.}.
More precisely, we show that, for all $0<p\leq 1$, solutions to the Fritz John-type equation~\eqref{fritz_john_intro} 
with metric $\mathbf{g}_p$ given by~\eqref{metric_intro0}, 
exist globally in time (to the future), provided that sufficiently small, regular and compactly supported initial data are prescribed. 


This work generalizes part of the results of \cite{jaj}, where \eqref{fritz_john_intro} were studied in the setting of \textbf{accelerated} FLRW spacetimes (in particular, $a(t)=t^p$, with $p > 1$). This work in the accelerated setting crucially exploited the integrability of $a^{-1}$ to Gr{\"o}nwall the main vector field estimates and to bootstrap energy boundedness. In the non-accelerated case that we analyze here, such an integrability condition is lacking and different methods are required. For instance, due to the non-compact nature of the spatial sections and the slower expansion rate, our analysis has to deal with the non-trivial effects that spacetime expansion has on dispersion.

In the following, we describe the main results and the techniques that we employ. 

\subsection{Preliminary version of the main results}

Given a spacetime $(\mathcal{M}, \metric)$, with $\mathcal{M}  \cong \mathbb{R}^{1+3}$ and $\metric = \metric_p$ given by~\eqref{metric_intro0}, so that $a(t)=t^p$, 
let $\phi \colon \mathcal{M} \to \mathbb{R}$ be sufficiently regular and let 
\[
E[\phi](t) = \frac12 \int_{\mathbb{R}^3} \br{ a^2(t) \br{ \partial_t \phi}^2(t, x) + |\nabla_x \phi|^2(t, x) } dx
\]
be its associated energy. Given $0 < t_0 < t_1$, we introduce the notation 
\[
\left \| \phi  \right \|^2_{L^2[t_0, t_1]} = \int_{t_0}^{t_1} \int_{\mathbb{R}^3} |\phi|^2(t, x) dx dt \, .
\]

\begin{theorem}[Improved energy estimate] \label{thm:intro_energyestimate}
    Let $0 < p \le 1$. Given $\phi \colon [t_0, t_1) \times \mathbb{R}^3 \to \mathbb{R}$ sufficiently regular, then:
    \[
        \sup_{t_0 \le t < t_1} E^{\frac12}[\phi](t) + C(p) \lp{t^{p-\frac12} \partial_t \phi}{L^2[t_0, t_1]} \lesssim_p E^{\frac12}[\phi](t_0) + \lp{t^{p+\frac12} \square_{\mathbf g} \phi}{L^2[t_0, t_1]} \, . 
    \]
\end{theorem}
Compared to the energy estimate obtained in \cite[Theorem 3.2]{jaj}, here we included the spacetime term $\| t^{p-\frac12} \partial_t \phi \|_{L^2[t_0, t_1]}$.  This term can be traced back to the spacetime expansion, note for instance that $C(p) = 0$, when $p=0$,  and will be exploited to absorb terms in the non-linear setting.

In this context, we introduce the set of vector fields
\begin{equation} \label{subalgebra}
\mathbb{L} = \{ \partial_{\tau}, \partial_i, S = \tau \partial_{\tau} + r\partial_r , \Omega_{ij} = x^i \partial_j - x^j \partial_i, \quad \text{with } \, i, j=1, \ldots, 3 \},
\end{equation}
which span a Lie algebra
with the usual bracket operation, where $\tau = \int \frac{dt}{a(t)}$ is a conformal time coordinate and $r^2=\sum_i(x^i)^2$. Vector fields belonging to $\mathbb{L}$ will be denoted by $\Gamma$. Then, for $k \ge 0$ and for a suitable norm $\| \cdot \|$ we also introduce the notation
\[
\left \| \phi_{\le k} \right \| = \sum_{|\alpha| \le k} \left \| \Gamma^{\alpha} \phi \right  \|.
\]

\begin{theorem}[Global pointwise estimate] \label{thm:intro_globalpointwise}
    Given a sufficiently regular $\phi \colon \mathcal{M} \to \mathbb{R}$, if $k \ge 0$ and $0 <p < 1$, we have:
    \begin{align*}
        \lp{t \partial_t \phi_{\le k}(t, \cdot)}{L^{\infty}(\mathbb{R}^3)} + \lp{ t^{1-p} \nabla_x \phi_{\le k} (t, \cdot)}{L^{\infty}(\mathbb{R}^3)} &\lesssim_{p, k} E^{\frac12}[\phi_{\le k+4}](t_0) + \\
        &+ \lp{ t^{p+\frac12} \br{\square_{\mathbf g} \phi}_{\le k+4}}{L^2[t_0, t_1]} \\
        &+ \lp{ t^{p+1} \br{\square_{\mathbf g} \phi}_{\le k +1}(t, \cdot)}{L^2(\mathbb{R}^3)} \, .
    \end{align*}
\end{theorem}

Finally, in the case in which $\phi$ solves \eqref{fritz_john_intro}, we prove the following:
\begin{theorem}[Small data global existence of solutions to the Fritz John equation in non-accelerated expanding spacetimes] \label{thm:SDGE}
    Let $\mathbf{g}$ be the FLRW metric \eqref{metric_intro} with scale factor $a(t) = t^p$, where $0 < p \le 1$. Let $(\phi_0, \phi_1)$ be sufficiently regular, small and compactly supported initial data. Then, given $t_0 > 0$, there exists a unique future \textbf{global} solution to the Cauchy problem
    \[
    \begin{cases}
        \square_{\mathbf g}  \phi = \br{ \partial_t \phi }^2, \\
        \phi(t_0, x) = \phi_0(x), \; \partial_t \phi(t_0, x) = \phi_1(x).
    \end{cases}
    \]
    Moreover, for every $0 \le k \le K - 4$, if $0 < p < 1$: 
    \[
    \lp{ t \partial_t \phi_{\le k} (t, \cdot )}{L^{\infty}(\mathbb{R}^3)} + \lp{ t^{1-p} \nabla_x \phi_{\le k}(t, \cdot)}{L^{\infty}(\mathbb{R}^3)} \leq 1/C_p, \quad \forall\, t \ge t_0 \,,
    \]
    where $C_p\rightarrow 0$, when $p\rightarrow 0$, and $ K\in\mathbb{Z}^+$, with $K\geq 8$, is the number of derivatives of the initial data assumed to be small.
    
    Furthermore, for every $0 \le k < K - \frac32$, if $p=1$:
    \[
    \lp{t \partial_x^k \partial_t \phi(t, \cdot)}{L^{\infty}(\mathbb{R}^3)} + \lp{\partial_x^k \partial_x \phi(t, \cdot)}{L^{\infty}(\mathbb{R}^3)}  \lesssim 1, \quad \forall\, t \ge t_0 \, .
    \]
\end{theorem}
In combination with the results of \cite{jaj}, we obtain the following.
\begin{corollary}[Small data global existence of solutions to the Fritz John equation in expanding spacetimes]
    Small data future global existence, in the sense of theorem \ref{thm:SDGE}, holds for the Fritz John equation
    \[
    \square_{\mathbf g}  \phi = (\partial_t \phi)^2,
    \]
    on the background FLRW spacetime, with  $\mathbb{R}^{+}\times\mathbb{R}^3$ topology and metric  
    \[
    \mathbf g  = -dt^2 + t^{2p} \br{(dx^1)^2 + (dx^2)^2 + (dx^3)^2},
    \]
    for every $p > 0$.
\end{corollary}

\subsection{Main techniques of the proof}
\label{secTechniques}
We prove physical space bounds that apply both to the linear and non-linear setting. In the linear case, these bounds can be interpreted, in a sense that will be made clear below,  either as dispersion-type estimates or as expansion-type estimates.  

In fact, let $\phi$ solve the \textbf{linear} wave equation $\square_{\mathbf g} \phi = 0$.
Then, the energy estimate in theorem \ref{thm:intro_energyestimate} allows to prove (after commuting with spatial derivatives and using Sobolev inequalities) a
\begin{equation} \label{expansion_type_est}
|\partial_t \phi| \lesssim \frac{1}{t^p}
\end{equation}
decay. We interpret this as an expansion-type of estimate
for the following reasons: Firstly, this time decay becomes negligible\footnote{Note that most of our results degenerate in the $p \to 0$ limit. This is expected due to the blow-up behaviour of the Fritz John equation in Minkowski spacetime. We emphasize this in the final statement of our results, see already section \ref{setup}.} in the $p$-small regime, similarly to the Minkowski ($p=0$) case, for which no expansion occurs. 
 Secondly, we can compare this decay to the case
when the spatial sections are compact and given by a flat torus, where dispersion is absent and the only decay mechanism is given by expansion.
In that case, the sharp decay for solutions to the linear wave equation is $|\partial_t \phi| \lesssim t^{-2p}$ when $0 < p < 1$ (see \cite[Appendix A]{jpj}).   One should also note that, 
for the $p>1$ case, when the expansion is stronger, these expansion-type estimates, together with the integrability of the inverse of the scale factor, are sufficient to prove global existence in the non-linear case~\cite{jaj}.

On the other hand, commutations with all $\Gamma \in \mathbb{L}$ and using generalized Klainerman-Sideris estimates
in combination with the estimate \eqref{expansion_type_est},
 yield additional pointwise decay of the form
    \begin{equation} \label{dispersion_type_estimate}
    |\partial_t \phi| \lesssim \frac{1}{t}
    \end{equation}
for solutions to the linear wave equation. We interpret
this bound as a consequence of the intertwining nature of dispersion and expansion since this requires the expansion-type estimate \eqref{expansion_type_est} and since, in the $p$-small regime, it can be compared to the 
global $t^{-1}$ decay for linear waves in the Minkowski ($p=0$) case. In the latter, the decay is entirely due to dispersion. 


To prove the energy (expansion-type) estimate, we exploit an energy current argument with vector field multiplier $X = a^{-1}\partial_t$. Differently from \cite{jaj}, we keep all terms containing the spacetime norm
$\| t^{p-\frac12} \partial_t \phi \|_{L^2[t_0, t_1]}$. This will be crucial to absorb terms in the non-linear setting.

To show the global pointwise estimate \eqref{dispersion_type_estimate}, our analysis relies on the following basic fact: given $[\square_{\mathbf g}, \Gamma]\phi$, where $\Gamma$ is a vector field of the Poincaré algebra,\footnote{That is, the Lie algebra spanned by $\mathbb{L}$ defined in \eqref{subalgebra}, and by the boosts vector fields $\Omega_{0j} = -\tau \partial_j-x^j \partial_{\tau}$, $j=1, 2, 3$.} this commutator has different long-time behaviours depending on the vector field $\Gamma$. Due to the time-dependence of the metric, the commutators coming from the boosts vector fields $\Omega_{0j} = -\tau \partial_j - x^j \partial_{\tau}$ give rise to slowly-decaying terms, compared to the other $\Gamma$ fields. In particular,
if one were to include the boosts, it is not clear how to generalize
the Klainerman-Sobolev estimates to the FLRW scenario in order to prove small data global existence, due to such error terms.

Here we bypass this problem by adopting a different approach: since the ``bad'' error terms originate from the $\Omega_{0j}$ vector fields, we exploit the good structure of the boosts-free subalgebra of the Poincaré algebra. This was already used by Klainerman and Sideris in \cite{kl_sid} to prove almost-global existence results for some quasilinear wave equations in $\mathbb{R}^{1+3}$. Our generalization of the Klainerman-Sideris estimates to the FLRW case, together with a careful computation of the commutation errors 
inside and outside the wave zone, allows to obtain the global pointwise
estimate described above.

At the nonlinear level, we exploit the decay estimate 
\eqref{dispersion_type_estimate} through a continuity argument. At first glance, one 
might worry that this global pointwise decay rate {\em{is still non-integrable in time}} 
and could therefore lead to finite-time blow-up, as in Minkowski space for the 
Fritz John equation. The key distinction in the FLRW setting with $0<p<1$
is that the expansion-type energy estimate additionally provides control of the spacetime norm $\| t^{p-\frac12} \partial_t \phi \|_{L^2[t_0, t_1]}$.
Even though this norm is not necessary to prove the linear estimates of the current paper, it is however needed in the non-linear bootstrap argument to control the non-linearities. 
 As a consequence, the bootstrap estimates can be closed without requiring time-integrable decay.
This control of non-linearities is  reminiscent
of the way in which the integrated local energy decay (ILED) spacetime
norm is used to control errors terms with time-weighted vector fields
within the asymptotically flat regime (see \cite{Olivernontrap} and \cite{OliverSterbenzBH}).

\subsection{Related works and outlook}
The Einstein equations, in harmonic coordinates, can be cast as a system of quasilinear wave equations for the metric tensor. A fundamental surrogate model to study the main stability features of metric solutions is the linear wave equation on a curved background. If $\metric$ is an FLRW metric, as in our case, explicit (Kirchhoff-type) expressions for solutions to 
\[
\square_{\metric} \phi = 0
\]
were derived for specific choices of the scale factor $a(t)$ and for different curvatures of the spatial sections (see \cite{KlainermanSarnak1981} and the more recent analyses in \cite{AbbasiCraig2014, NatarioRossetti2023}). In this context, energy methods proved effective to derive physical-space estimates of solutions and of their derivatives (see \cite{jpj} for the $p > 1$ case, and \cite{Haghshenas2025} for the $ 0 < p < 1$ case). We remark that a physical-space proof yielding sharp global decay rates of solutions is still lacking, whereas a Fourier-based proof can be derived from the results in \cite{Wirth2004}.  Already at the linear level, the main difficulties of the analysis of waves originates from the non-stationarity of such spacetimes and from the fact that they are not asymptotically flat (although they are conformally flat, see the conformal time coordinate introduced in \eqref{subalgebra}).

The analysis of solutions to
\begin{equation} \label{nonlinear_intro}
\square_{\metric} \phi = \mathcal{N}(\phi, \phi’),
\end{equation}
where $\mathcal{N}$ encodes non-linear terms, is relevant both to understand the impact of a dynamical background on non-linear effects, and also as a step towards the study of the stability of solutions to the Einstein-matter equations. 

When $\mathcal{N}$ contains only derivatives of solutions, there exist a structure for $\mathcal{N}$, inspired by the decay of linear waves, that guarantees small data global existence when the underlying FLRW spacetimes undergo an accelerated expansion (in particular, the $p>1$ case of our setup) \cite{jaj}. Such a non-linear structure  includes the case $\mathcal{N} = (\partial_t \phi)^2$ that we study in the present work. Therefore, our work generalizes the results of \cite{jaj} for $\mathcal{N} = (\partial_t \phi)^2$ to the case of non-accelerated spacetimes ($0<p\leq 1$).  We focus on this type of non-linear term due to its relevance in the analysis of wave equations in $\mathbb{R}^{1+3}$. However,  our methods are general and theorem \ref{thm:SDGE} is the only result of the present work that relies on this specific form of the non-linear terms.

Equation \eqref{nonlinear_intro} is also studied in the context of wave equations in $\mathbb{R}^{1+3}$ with time-dependent (scale-invariant) damping terms (see \eqref{damped_eqn}). In \cite{TsutayaWakasugi2021} (see also \cite{WeiYong2024}), it was proved that, in the decelerating regime $0 < p < 1$, there exist non-trivial solutions to \eqref{nonlinear_intro} with $\mathcal{N} = |\partial_t \phi|^q$ that blow up in finite time when $1 < q \le 1+ \frac{1}{1+p} < 2$. In \cite{TW2}, the authors proved small data global existence when $q > 1 + \frac{1}{p}$. This is related to the Glassey conjecture in FLRW spacetimes. Our work shows that $q = 1 + \frac{1}{p}$ is not the critical exponent when $0 < p \le 1$.

 In \cite{PT}, blow-up results for $1 < q \le f(p) < 2$, for a function $f(p)$, are proved. Blow-up of solutions was later proved for $q \le 2-p$ in \cite{HHP}. In \cite{FH}, small data global existence is proved for $q \ge 3$ and the result in the range $2-p < q < 3$ is deemed open (see their remark 3.3). Our result shows that small data global existence holds for $q=2$ as well.

We refer to the review in \cite{TsutayaWakasugi2022} for an overview on the case $\mathcal{N}=\mathcal{N}(\phi)$, which is related to the generalization of the Strauss conjecture to cosmological spacetimes.

Furthermore, we point out the numerical work \cite{fa}, where the long-time decay of linear waves and derivatives in FLRW spacetimes was investigated, and \cite{RossettiVanoVinuales2025}, where the analysis was repeated by employing more robust numerical methods (involving a hyperboloidal foliation of the spacetime) and where non-linear terms were also included. In the latter, it was suggested that small data global existence may actually have failed for the Fritz John equation \eqref{fritz_john_intro} for very slow spacetime expansions. This was put in the context of the challenging problem of studying non-linear instabilities from a numerical point of view. The present paper disproves such a suggestion: we show (see theorem \ref{thm:SDGE}) that small data global existence holds for the metric given by \eqref{metric_intro}, \eqref{scale_factor_intro} for every positive value of $p$. 

The future stability of FLRW metrics has been investigated in the case of the Einstein-massless Vlasov system. This requires a radiation-type of fluid ($p=1/2$ in \eqref{scale_factor_intro}), see \cite{TaylormVlasov} in the spherically symmetric setting and \cite{TaylorVelozoVlasov} for an analysis of the Vlasov equation on fixed FLRW backgrounds. When $1/3 \le p \le 2/3$,\footnote{This assumption is required so that the square of the speed of sound of the fluid lies in $[0, 1]$.} FLRW metrics solve the Einstein-Euler system of equations (see \cite{fluidsReview} for a recent review and \cite{RodnianskiSpeck, Speck,Oliynyk} for related works). Similarly to the case of the present work on the wave equation, expansion provides a regularizing effect also for relativistic fluids \cite{SpeckFixedBackground, FajmanOlinykWyattNonAccelerated, FajmanOfnerOliynykWyatt, phase_transition_fluids}. The FLRW metrics can also be interpreted as solutions of the Einstein-scalar field or Einstein-non linear scalar field system, see,  e.g.\ \cite{RingstromNonLinear, FournScalarField} and the recent \cite{Bernhardt}.

In the setting of the present paper, we focus on a specific semilinear wave equation as a model to understand the regularizing effects of spacetime expansion. We leave the task of finding a null condition in the non-accelerated setting and generalizing our methods to quasilinear systems for future work.

\subsection{Outline of the paper}
In section \ref{setup} we present the main problem and state the rigorous version of our results. In section \ref{secEnergy} we derive the main energy estimate that we will use throughout the paper. This estimate includes a weighted spacetime term coming from expansion. In section \ref{secCommutators} we analyse the commutation relations of the Laplace-Beltrami operator in expanding FLRW spacetimes, with respect to the generators of the Poincaré algebra. In section \ref{secPointwise} we obtain the pointwise (Klainerman-Sideris) decay estimates by means of commutation with a boosts-free subset of the Poincaré algebra. In section \ref{secGlobal} we show the small data global existence result in the specific case of the Fritz John equation $\square_g \phi = (\partial_t \phi)^2$, when $0 < p < 1$, by exploiting the previous energy and pointwise bounds and a bootstrap argument.  This result is then proved in the case $p=1$, without need of commuting, in section \ref{secp1}.

\section{Setup and Main Results}
\label{setup}
%
Let $(\mathcal{M},{\mathbf{g}})$
be a Friedmann-Lemaitre-Robertson-Walker (FLRW) spacetime
with topology $\mathbb{R}^+\times\mathbb{R}^3$ and metric
\begin{equation}
\label{metricFLRW}
   {\mathbf{g}} := -dt^2 + t^{2p}\delta_{ij}dx^idx^j\;,  
\end{equation}
where $0<  p \le 1$ and $\delta_{ij}$ denotes the Kronecker delta.
Our choice of the scale factor
\begin{equation}
\label{power_law}
 a(t)=t^p,
\end{equation}
%
with $0<p\leq 1$, implies that the cosmologies of the present work undergo a {\bf{non-accelerated}}
power law expansion in the direction of positive time $t$. A relevant consequence of the previous choices is that
\begin{equation}
\label{non_accel}
\lim_{T\rightarrow \infty} \int_{t_0}^{T}  \frac{1}{a(s)} ds=
\lim_{T\rightarrow \infty}\int_{t_0}^{T}  \frac{1}{s^p} ds= \infty ,
\qquad 0< p \le 1\;,
\end{equation}
where, from now on,  $t_0>0$ is fixed. Consider the covariant wave operator
\begin{equation}
\square_{\mathbf{g}}\phi = \frac{1}{\sqrt{|\mathbf{g}|}}
\partial_{\alpha}\left(\mathbf{g}^{\alpha\beta}
\sqrt{|\mathbf{g}|} \partial_{\beta}\phi\right ) \;,
\end{equation}
where $|\mathbf{g}|=-\det(\mathbf{g}_{\alpha\beta})$, $\mathbf{g}^{\alpha\beta}$
are the components of the inverse of $\mathbf{g}_{\alpha\beta}$, and where
greek indices run from $0$ to $3$. For the FLRW metric~\eqref{metricFLRW},we have
\bea
\lb{flrw_eq}
\square_{\mathbf{g}}\phi=-\partial_t^2 \phi
-\frac{3 p}{t}\partial_t \phi+ \frac{1}{t^{2p}}\Delta \phi\;,
\eea
where $\Delta$ is the Laplace operator on $\mathbb{R}^{3}$. In terms of the \textit{conformal time coordinate} $\tau$, defined up to a constant by 
\begin{equation} \label{def_tau}
\tau = \int \frac{dt}{a(t)} \; ,
\end{equation}
the FLRW metric becomes
\[
\mathbf{g} = a^2(\tau) \left ( -d\tau^2 + \br{dx^1}^{2} + \br{dx^2}^2 + \br{dx^3}^2 \right )
\]
and the respective covariant wave operator reads
\begin{equation} \label{Laplace_Beltrami_FLRW}
\square_{\mathbf{g}} \phi = -\frac{1}{a^2}\partial_{\tau}^2 \phi -
2\frac{\partial_{\tau} a}{a^3} \partial_{\tau} \phi + \frac{1}{a^2}\Delta \phi \; .
\end{equation}
 In the following, we choose the additive constant in the definition of conformal time so that
\begin{equation} \label{conf_normalization}
\tau = \frac{t^{1-p}}{1-p}.
\end{equation}%
\subsection{Vector fields and notation} \label{section:notation}
Next, we define the set\footnote{A quick computation shows that the span of $\mathbb{L}$ is a Lie subalgebra of the Poincaré algebra spanned by $\mathbb{P} = \{ \partial_{\tau}, \partial_i, S, \Omega_{ij}, \Omega_{0j}\}$, with $\Omega_{0j} = -\tau \partial_j - \delta_{jk }x^k\partial_{\tau}$ and $i, j = 1, 2, 3$.}
of vector fields in $(\tau,x^i)$ coordinates:
    \begin{equation*}
        \mathbb{L} \  =  \ \big\{  \partial_\tau= a\partial_t
        ,\ \partial_i=\partial_{x^i} , \  
        S=\tau\partial_{\tau} + r\partial_r  , \ 
        \Omega_{ij}=x^i\partial_j-x^j\partial_i
        \big\} \ , \label{mod_fields}
    \end{equation*}
where $1\leq i\leq 3$. For translational derivatives we also define
    \[\nabla_{\tau, x}\phi=(\partial_{\tau}\phi, \partial_1\phi, \partial_2\phi ,\partial_3\phi)\ ,
    \qquad \nabla_{t, \tilde x}\phi=(\partial_{t}\phi,
    a^{-1}\partial_1\phi, a^{-1}\partial_2\phi ,a^{-1}\partial_3\phi) \ , \]
where we denote the {\it{unit normalized}}
spatial derivatives by
\[\tilde \partial_i:=a^{-1}\partial_{x^i} \ .\]
Given a multi-index of $1+3$ non-negative integers
$\alpha=(\alpha_0,\alpha_1, \alpha_2 ,\alpha_{3})$ we define $\nabla^{\alpha}_{\tau, x} \phi$ as the vector of ${|\alpha| + 3 \choose 3}$ different components, where each component can be expressed as
    \[ \partial_{\tau}^{\alpha_0} \partial_1^{\alpha_1} \partial_2^{\alpha_2} \partial_3^{\alpha_3} \phi, \]
where again $\alpha_0 + \alpha_1 + \alpha_2 + \alpha_3 = |\alpha|$.
For a triplet $\alpha=(i,j,k)$ of
multi-indices $i,j$ and a nonnegative integer $k$
we denote
$|\alpha|=|i|+|j|+k$ and
\[\phi_{\alpha}=\nabla_{\tau,x}^{i}\Omega^{j}S^{k}\phi \ ,\]
where $\Omega:=(\Omega_{ij})$ with $1\leq i<j\leq 3$
denote the spatial rotations.
Let $f:\mathbb{R}^{+}\times \mathbb{R}^{n}\rightarrow \mathbb{R}$
have sufficiently fast decay
at spatial infinity. We use the norms
\[\lp{f}{L^2}:=
\lp{f(t,\cdot)}{L^2}=\left(\int_{\mathbb{R}^{3}}\left|f(t,x)\right|^2  dx\right)^{\frac12}\;,
\qquad \lp{f}{L^2[t_0,t_1]}:=\left(\int_{t_0}^{t_1}
\!\!\!\int_{\mathbb{R}^{3}}\left|f(t, x)\right|^2  dxdt\right)^{\frac12}\;. \]
Combining these definitions, for $n \ge 0$ we write
\begin{equation} \label{def_len} 
\lp{\Gamma^{\le n} \phi }{L^2} :=
\lp{\phi_{\leq n}}{L^2}:=\sum_{0 \le |\alpha|\leq n}\lp{\phi_{\alpha}}{L^2}
=\sum_{0 \le |\alpha|\leq n}\lp{\nabla_{\tau,x}^{i}\Omega^{j}S^{k}\phi}{L^2} \ ,
\end{equation}
and we use an analogous notation for combinations of vector fields, for instance:
\[
\lp{ \Gamma \, \Gamma^{\le n} \phi}{L^2} \coloneqq \sum_{0 \le |\alpha|\leq n}\lp{ \Gamma \nabla_{\tau,x}^{i}\Omega^{j}S^{k}\phi}{L^2}
\]
for any $ \Gamma \in \mathbb{L}$ 
and, due to the commutation relations of $\mathbb{L}$, we will see that
\[
\lp{\Gamma^{\le m} \Gamma^{\le n} \phi}{L^2} \sim \lp{\Gamma^{\le m + n}  \phi}{L^2} \; ,
\]
for any $n, m \ge 0$.
In a similar way, we also define the following notation for
the commutators of $\square_{\mathbf{g}}$ with products
of vector fields $\Gamma\in \mathbb{L}$:
\[\lp{[\square_\mathbf{g},\Gamma^{\leq n}]\phi}{L^2}
:=\sum_{|\alpha|\leq n}\lp{[\square_{\mathbf{g}},\Gamma^{\alpha}]\phi}{L^2}
=\sum_{|\alpha|\leq n}\lp{[\square_{\mathbf{g}},\nabla_{\tau,x}^{i}
\Omega^{j}S^{k}]\phi}{L^2} \ ,\]
with analogous definitions for the spacetime norms $L^{2}[t_0,t_1]$ and for the supremum norm. We stress that the commutation relations between the above vector fields and $\square_{\mathbf{g}}$ will play a paramount role. For instance, given two integers $n$ and $m$, the relation between $| (\square_g \phi_{\le n})_{\le m}|$ and $|(\square_g \phi)_{\le n + m}|$ is non-trivial and will be analyzed in detail (see already corollary \ref{coroll:switch_brackets}).

\subsection{Main Results I: Linear Estimates}
\label{subsec_main_linear}
In this section we collect our main $L^2$ and $L^{\infty}$ estimates providing control of general smooth functions in $\mathbb{R}^{+}\times\mathbb{R}^3$, with spatial compact support, involving derivatives in $\mathbb{L}$ and the Laplace-Beltrami operator of FLRW metrics:
\begin{theorem}
\label{energyEst_Gamma_many_th}
Let $\mathbf{g}$ denote a smooth metric
of the form~\eqref{metricFLRW} with $0<p<1$, describing an FLRW spacetime
with topology $\mathbb{R}^{+}\times\mathbb{R}^3$. Let $0< \tau_0< \tau_1$ and assume that $\phi(\tau, \cdot) \in C^{\infty}_c(\mathbb{R}^3)$ for every $\tau \in (\tau_0, \tau_1)$.
For any fixed $k\geq 0$,
there exists a constant $C:=C(p,k)>0$ such that
the following estimate holds in $(\tau,x)$ coordinates:
    \begin{align}
    \sup_{\tau_0\leq \tau \leq  \tau_1}
    \lp{\nabla_{\tau, x}\phi_{\leq k}(\tau, \cdot)}{L^2}+
    C\lp{\tau^{-\frac{1/2}{1-p}}\partial_{\tau}\phi_{\leq k}}
    {L^2[\tau_0,\tau_1]}&\lesssim_p  
    \lp{\nabla_{\tau, x}\phi_{\leq k}(\tau_0, \cdot)}{L^2}
   +\lp{\tau^{\frac{p+1/2}{1-p}}
    (\square_{\mathbf{g}}\phi)_{\leq k}}{L^2[\tau_0,\tau_1]} \, ,
     \label{energyEst_Gamma_many}
\end{align}
where the constants depending on $p$ degenerate in the limit $p \to 0$.
Equivalently, for the unit normalized components of the gradient
    in $(t,x)$ coordinates we have: 
       \begin{align}
    \sup_{t_0\leq t \leq  t_1}
    \lp{t^{p}\nabla_{t, \tilde x}\phi_{\leq k}(t, \cdot)}{L^2}+
    C\lp{t^{p-1/2}
    \partial_{t}\phi_{\leq k}}
    {L^2[t_0,t_1]}&\lesssim_p  
    \lp{t^{p}\nabla_{t , \tilde x}\phi_{\leq k}(t_0, \cdot)}{L^2}
   +\lp{t^{p+1/2}
    (\square_{\mathbf{g}}\phi)_{\leq k}}{L^2[t_0,t_1]}
     \label{energyEst_Gamma_many_tx} \, .
\end{align}
\end{theorem}
When $p=1$, the case $k=0$ (no commutations with the vector fields in $\mathbb{L}$) is sufficient to prove small data global existence.
\begin{corollary}
Let $\mathbf{g}$ denote a smooth metric
of the form~\eqref{metricFLRW} with $p=1$, describing an FLRW spacetime
with topology $\mathbb{R}^{+}\times\mathbb{R}^3$. Let $0<t_0<t_1$ and assume that $\phi(t, \cdot) \in C^{\infty}_c(\mathbb{R}^3)$ for every $t \in (t_0, t_1)$. Then, there exists a constant $C:=C(p)>0$ such that
the following estimate holds:
\begin{align}
    \sup_{t_0\leq t \leq  t_1}
    \lp{t \nabla_{t, \tilde x}\phi (t, \cdot)}{L^2}+
    C\lp{t^{1/2}
    \partial_{t}\phi}
    {L^2[t_0,t_1]}&\lesssim  
    \lp{t \nabla_{t , \tilde x}\phi(t_0, \cdot)}{L^2}
   +\lp{t^{3/2}
    \square_{\mathbf{g}}\phi}{L^2[t_0,t_1]}
     \, .
\end{align}
\end{corollary}

Using the above estimates, we produce a global 
pointwise decay result for the gradient of solutions to the linear wave
equations propagating in our family of metrics. The following commuted estimate will only be exploited in the case $0 < p < 1$.

\begin{theorem}(Global pointwise decay for $\nabla_{\tau, x} \phi$)
\label{global_ptwise_decay_linear}
    Let $\mathbf{g}$ denote a smooth metric
of the form~\eqref{metricFLRW} with $0<p<1$, describing an FLRW spacetime
with topology $\mathbb{R}^{+}\times\mathbb{R}^3$. Let $0< \tau_0< \tau_1$ and assume that $\phi(\tau, \cdot) \in C^{\infty}_c(\mathbb{R}^3)$ for every $\tau \in (\tau_0, \tau_1)$. Then, for any fixed $k\geq 0$, 
    in $(\tau,x)$ coordinates, we have
    \begin{align}
    \label{global_ptdecay_grad1}
     \lp{\tau \nabla_{\tau, x}\phi_{\leq k}(\tau, \cdot)}{L^{\infty}}
     \lesssim_{p,k}\lp{\nabla_{\tau, x}\phi_{\leq k+4}(\tau_0, \cdot)}{L^2}
     &+\lp{\tau^{\frac{p+1/2}{1-p}}
     (\square_{\mathbf{g}}\phi)_{\leq k+4}}{L^2[\tau_0,\tau_1]}\\
     &\hspace{1in} +\lp{\tau^{1+\frac{2p}{1-p}}
     \big(\square_{\mathbf{g}}\phi\big)_{\leq k+1}(\tau, \cdot)}{L^2} \ ,
     \notag 
    \end{align}
    where the constants depending on $p$ degenerate in the limit $p \to 0$.
    Equivalently, for the unit normalized components of the gradient
    in $(t,x)$ coordinates we have
    \begin{align}
    \label{global_ptdecay_grad2}
      \lp{t \nabla_{t,\tilde x}(\phi_{\leq k})(t, \cdot)}{L^{\infty}}
     \lesssim_{p,k}\lp{t^{p}\nabla_{t, \tilde x}(\phi_{\leq k+4})(t_0, \cdot)}{L^2}
     &+\lp{t^{p+1/2}
     (\square_{\mathbf{g}}\phi)_{\leq k+4}}{L^2[t_0,t_1]}\\
     &\hspace{1in}+\lp{t^{p+1}
     \big(\square_{\mathbf{g}}\phi\big)_{\leq k+1}(t, \cdot)}{L^2}
      \ .  \notag
    \end{align}
\end{theorem}
\subsection{Main Results II: Nonlinear Estimates}
\label{subsec_main_nonlinear}
The rest of our paper
concerns the following
Cauchy problem for the semilinear wave equation
\begin{equation}
\label{FJequation}
\begin{cases}
\ \square_{\mathbf{g}}\phi=(\partial_t\phi)^2\;, \\
\displaystyle \ \phi(t_0,x)=\phi_0(x) \;\;, \; \partial_t\phi(t_0,x)=\phi_1(x) \;.
\end{cases}
\end{equation}
This is related to one of the examples of small-data finite time blow up
investigated by F. John in $\mathbb{R}^{1+3}$ \cite{john_blow}. In contrast
to the case of Minkowski space, here we prove small data future global
well-posedness in our class of FLRW spacetimes:
\begin{theorem}
{(Small data global wellposedness for the semilinear
Fritz John equation)}
\label{main1}
Let $K\geq 8$ and consider an FLRW spacetime
with topology $\mathbb{R}^{+}\times\mathbb{R}^3$ and smooth metric
of the form~\eqref{metricFLRW} with $0<p \le 1$. 
Consider 
$\phi_0,\phi_1:\mathbb{R}^3\rightarrow\mathbb R$ with $\phi_0, \phi_1 \in C^{\infty}_c(\mathbb{R}^3)$, satisfying the condition
\begin{equation}
\label{smallInitial_data_FJ}
{\left(\lp{(\phi_0)_{\leq K+1}}{L^{2}}  +\lp{(\phi_1)_{\leq K}}{L^{2}}\right)=:C_0<\infty\;.}
\end{equation}

Then, given $t_0>0$, there exists $\delta_0>0$, such that,
if $C_0 \leq \delta_0$, the initial value problem
	\eqref{FJequation}
	admits a unique smooth and compactly supported global solution $\phi$, i.e.\ $\phi(t, \cdot) \in C^{\infty}_c(\mathbb{R}^3)$ for every $t \ge t_0$.

Furthermore, for every {$0\leq k \le K-4$}, if $0 < p < 1$:
\[\lp{t \nabla_{t,\tilde x}\phi_{\leq k}(t,\cdot )}{L^{\infty}}
\le C(p)\;, \quad \forall \, t \ge t_0 \, ,
\]
where $C(p) \to \infty$ when $p \to 0$.

Moreover, for every $0 \le k < K - \frac32$, if $p=1$:
\[\lp{t \partial_x^k \nabla_{t,\tilde x}\phi(t,\cdot )}{L^{\infty}}
\lesssim  1 \;, \quad \forall \, t \ge t_0\, .
\]
\end{theorem}
\begin{remark}[On regularity]
     Our methods can be generalized to lower regularity settings, by defining the space of solutions, for $n \in \mathbb{Z}^+$ (see e.g.\ \cite{kl_sid}):
    \[
        H^n_{\Gamma}(\mathbb{R}^3) \coloneqq \left \{ f \colon \lp{\Gamma^{\le n} f}{L^2(\mathbb{R}^3)} = \lp{f_{\le n}}{L^2(\mathbb{R}^3)} < +\infty \right \}
    \]
    and construct solutions $\phi$ to the Cauchy problem such that 
    \[
    (\phi, \partial_t \phi) \in L^{\infty}([t_0, +\infty), H^{K+1}_{\Gamma}(\mathbb{R}^3)) \times L^{\infty}([t_0, +\infty), H^{K}_{\Gamma}(\mathbb{R}^3)).
    \]
    
    This requires additional straightforward, but technical, details to prove local existence of solutions and the persistence of regularity in such spaces. Since the main ideas of our work do not depend on the regularity of solutions, smoothness is assumed.
\end{remark}
\begin{remark}
The work \cite{fa} produces
numerical results for the pointwise decay
of spherically symmetric solutions
to the linear wave equation in non-accelerated  FLRW spacetimes.
In that work,
$\partial_t\phi$ is shown to decay pointwise at a faster rate than the ones obtained using our methods (that also hold for linear waves).
Therefore, for the linear
wave equation, the decay rates for $\partial_t\phi$ provided by our work are not optimal.
Nonetheless our methods do not require sharp decay to
establish small data global existence for
\eqref{FJequation}. 
\end{remark}
%
\section{Energy Formalism}
\label{secEnergy}
%
We define the {\it{Energy-Momentum Tensor}} to be
\begin{equation}
    T_{\alpha\beta} = \partial_\alpha \phi\, \partial_\beta \phi
    - \frac{1}{2}\mathbf{g}_{\alpha\beta}
    \partial^\mu \phi\,  \partial_\mu \phi  \;.
\end{equation}
Let $D$ be the Levi-Civita connection of the metric $\mathbf{g}$.
The divergence of the energy momentum is then
\[D^\alpha T_{\alpha\beta}=  \partial_{\beta}\phi\square_{\mathbf{g}}\phi\; .\] 
Given a (smooth) 
vector field $X$ we define the 1-form
\[{}^{(X)}P_\alpha = T_{\alpha\beta}X^\beta .\]
Taking its divergence yields
\begin{equation} \label{divergence}
    D^\alpha {}^{(X)}P_\alpha = \frac{1}{2}\;^{(X)}\pi^{\alpha\beta}T_{\alpha\beta}+X\phi\cdot  \square_{\mathbf{g}}\phi\;,
\end{equation}
where 
\[^{(X)}\pi_{\alpha\beta}:= {\mathcal L}_X \mathbf{g} _{\alpha\beta}
= D_{\alpha}X_{\beta}+D_{\beta}X_{\alpha} \;,\] 
is a symmetric 2-tensor known as the \textit{Deformation Tensor}  
 of
$\mathbf{g}$ with respect to $X$. Integrating the divergence identity~\eqref{divergence} over the time slab
$\{(t,x) \;|\;t_0\leq t\leq t_1\}$ and using Stokes'
theorem we get the following vector field {\it{Multiplier Identity}}
\begin{equation} \label{multident}
\begin{split}
    \int_{t=t_0}\; ^{(X)}P_\alpha N^\alpha |\mathbf{g}|^{\frac{1}{2}}dx -
    \int_{t=t_1}\; ^{(X)}P_\alpha N^\alpha |\mathbf{g}|^{\frac{1}{2}}dx 
    = \int_{t_0}^{t_1}\!\!\!\int_{\mathbb{R}^3}\left(\frac{1}{2} \;^{(X)}\pi^{\alpha\beta}T_{\alpha\beta}+
    X\phi \cdot \square_{\mathbf{g}}\phi\right)|\mathbf{g}|^{\frac{1}{2}}dxdt\;,
\end{split}
\end{equation}
where $N=\partial_t$ is the future pointing unit normal to the time slices
$t=const$, with $|\mathbf{g}|=-\det(\mathbf{g}_{\alpha\beta})=a^{6}$ and $dx=dx^1 dx^2 dx^3$. 

The integrand $^{(X)}P_\alpha N^\alpha$ in \eqref{multident} is the {\it{Energy Density}}
associated to the vector field multiplier $X$. We
then have the following expression for the {\it{Energy}} of $\phi$ if we take
$X=a^l \partial_t = a^{-1} \partial_t$:
\begin{equation}
\label{energyDef}
E[\phi](t):= \int_{\{t\} \times \mathbb{R}^3}\!\! (
^{(X)}\! P_\alpha N^\alpha |\mathbf{g}|^{\frac{1}{2}})dx=
\frac12\int_{\mathbb{R}^3}a^2(t)
\left[(\partial_{t}\phi)^{2}+\delta^{ij}\tilde \partial_{i}\phi
\tilde \partial_{j}\phi\right] (t,x)
dx =\frac{1}{2}\lp{a(t)\nabla_{t,\tilde x}\phi(t,\cdot)}{L^{2}}^2\;,
\end{equation}
where the unit normalized spatial derivatives $\tilde \partial_i$ were introduced in section \ref{section:notation}.
Alternatively, in $(\tau, x)$ coordinates:
\begin{equation}
\label{energyDef_tp_tau}
E[\phi](t)=
\frac12\int_{\mathbb{R}^3}
\left[(\partial_{\tau}\phi)^{2}+\delta^{ij} \partial_{i}\phi
 \partial_{j}\phi\right] (t,x) \ dx
 =\frac{1}{2}\lp{\nabla_{\tau,x}\phi(t,\cdot)}{L^{2}}^2 \;.
\end{equation}
Recall the following identity in \cite{jaj}.
\begin{lemma}
For $X=a^l \partial_t$, where $a$ is the expanding factor and $l\in\mathbb{R}$, we have
\begin{equation}
    \;^{(X)}\pi^{\alpha\beta}T_{\alpha\beta}=a^{l-1}\dot a
    \left[(3-l) \left(\partial_t \phi\right)^2
    - (1 + l) \,\delta^{ij} \tilde \partial_i\phi \tilde \partial_j \phi
    \right]\;,\label{def_tensor_explicit}
\end{equation}
where $\delta^{ij}=:\delta_{ij}$ is the Kronecker delta.
\end{lemma}   
\noindent Using these identities
we can obtain our first energy estimate in the case $a(t)=t^p$. 
\begin{theorem}
(Improved Energy Estimate)
 Let $\mathbf{g}$ be the FLRW metric~\eqref{metricFLRW}
 and $\phi: [t_0,t_1) \times \mathbb{R}^{3}\rightarrow \mathbb{R}$
 be a $C^{2}$ function. The following weighted
 energy estimate then holds:
 \begin{equation} 
 \label{energyEst}
 \sup_{t_0\leq t< t_1}\lp{t^p\nabla_{t,\tilde x}\phi(t, \cdot)}{L^{2}}
 + (2p)^{1/2}\lp{t^{p-1/2}\partial_t\phi}{L^2[t_0,t_1]}
\lesssim_p \lp{t^p\nabla_{t,\tilde x}\phi(t_0, \cdot)}{L^{2}}+ \lp{t^{p+1/2}
\square_{\mathbf{g}}\phi}{L^2[t_0,t_1]}\;, \end{equation}
where the constants at the right hand side degenerate as $p \to 0$.
Alternatively, in $(\tau, x)$ coordinates:
 \begin{equation} 
 \label{energyEst_tau}
 \sup_{t_0\leq t< t_1}\lp{\nabla_{\tau, x}\phi(t, \cdot)}{L^{2}}
 + (2p)^{1/2}\lp{t^{-1/2}\partial_{\tau}\phi}{L^2[t_0,t_1]}
\lesssim_p \lp{\nabla_{\tau, x}\phi(t_0, \cdot)}{L^{2}}+ \lp{t^{p+1/2}
\square_{\mathbf{g}}\phi}{L^2[t_0,t_1]} \;.
  \end{equation}
In particular, $C=C(p)=(2p)^{1/2}$
in estimate \eqref{energyEst_Gamma_many}.

\end{theorem}
\begin{proof}
We choose $l=-1$. Substituting this along with
$a(t)=t^{p}$, $|\mathbf{g}|^{\frac{1}{2}}=t^{3p}$
into equations \eqref{energyDef} and
\eqref{def_tensor_explicit} yields
\begin{equation}
\label{energyDef_tp}
E[\phi](t)=
\frac12\int_{\mathbb{R}^3}t^{2p}
\left[(\partial_{t}\phi)^{2}+\delta^{ij}\tilde \partial_{i}\phi
\tilde \partial_{j}\phi\right] (t,x) \ dx
=\frac{1}{2}\lp{t^{p}\nabla_{t,\tilde x}\phi(t,\cdot)}{L^{2}}^2 
\end{equation}
and
\begin{equation}
\label{goodSign}
 |\mathbf{g}|^{\frac{1}{2}} 
 \big({}^{(X)}\pi^{\alpha\beta}T_{\alpha\beta}\big)=4p \,
 t^{2p-1}(\partial_t\phi)^2 \; ,
\end{equation}
respectively. Combining this with identity~\eqref{multident}, rearranging
terms, and taking absolute value yields
\begin{equation}
\label{energyEst0}
E(t_1) + 2p \int_{t_0}^{t_1}\int_{\mathbb{R}^3}
t^{2p-1}(\partial_t\phi)^2  \ dxdt
\leq E(t_0) + \int_{t_0}^{t_1}\int_{\mathbb{R}^3}
 |t^{2p}  \partial_t\phi \, \square_{\mathbf{g}}\phi|
 \,  dxdt \;.
\end{equation}
By Young's inequality with $\epsilon$:
\[\int_{t_0}^{t_1}\int_{\mathbb{R}^3}
 t^{2p}  |\partial_t\phi \, \square_{\mathbf{g}}\phi|
 \,  dxdt\leq \epsilon \int_{t_0}^{t_1}\int_{\mathbb{R}^3}
 t^{2p-1}  |\partial_t\phi|^2 \,
 \,  dxdt +\frac{1}{\epsilon }\int_{t_0}^{t_1}\int_{\mathbb{R}^3}
 t^{2p+1}|\square_{\mathbf{g}}\phi|^2
 \,  dxdt \]
Choosing $\epsilon > 0$ sufficiently small (depending on $p$), rearranging terms,
taking square roots and simplifying finishes the proof
of Eq. \eqref{energyEst}. In particular, the constants depending on $\epsilon$ (hence depending on $p$) are absorbed in the $\lesssim$ sign. The estimate
\eqref{energyEst_tau} follows immediately
from $t^p\partial_t=\partial_\tau$.
\end{proof}

%
\section{Commutators}
\label{secCommutators} 
%
In this section we focus on the case $0 < p < 1$ in \eqref{scale_factor_intro}.
We recall the expression \eqref{Laplace_Beltrami_FLRW} for the covariant wave operator expressed in conformal time:
\[
\square_{\mathbf{g}} \phi = -\frac{1}{a^2}\partial_{\tau}^2 \phi -
2\frac{\partial_{\tau} a}{a^3} \partial_{\tau} \phi + \frac{1}{a^2}\Delta \phi \; .
\]
Let:
\[
\begin{cases}
    S = \tau \partial_{\tau} + x^i \partial_i, \\
    \Omega_{ij} = \delta_{ik}x^k \partial_j - \delta_{jk}x^k \partial_i, \\
    \Omega_{0j} = -\tau \partial_j - \delta_{jk} x^k \partial_{\tau},
\end{cases}
\]   
where $\delta$ is the Kronecker delta.
Since $x^k$ is a scalar function for every $k$, in the following we will write $x_k$ to denote $\delta_{jk} x^j$, by a small abuse of notation. 
\begin{proposition}[Commutation relations associated to the generators of the Poincaré algebra] \label{prop:commutation_relations}
    Let $0 < p < 1$. Then, we have:
    \begin{align}
    [\square_{\mathbf{g}}, \partial_i] &= 0,  \qquad \text{ for } i=1, 2, 3\; ,  \label{commutator_di}\\
    [\square_{\mathbf{g}}, \Omega_{ij}] &= 0,  \qquad \text{ for } i, j = 1, 2, 3\; , \\
    [\square_{\mathbf{g}}, \partial_{\tau}] &= 2 \frac{\partial_{\tau}a}{a}\square_{\mathbf{g}} + 2 \frac{a \partial_{\tau}^2 a - (\partial_{\tau} a)^2}{a^4} \partial_{\tau} \; , \label{commutator_dtau}  \\
    [\square_{\mathbf{g}}, \Omega_{0j}] &= -x_j [\square_{\mathbf{g}}, \partial_{\tau}] + 2
    \frac{\partial_{\tau}a}{a^3}\partial_j, \label{commutator_boost} \qquad \text{ for } j=1, 2, 3 \; ,  \\
    [\square_{\mathbf{g}}, S] &=2 \left( 1 + \tau \frac{\partial_{\tau}a}{a}  \right) \square_{\mathbf{g}} .   \label{commutator_scaling} 
\end{align}
Furthermore, the above coefficients satisfy the following relations:
\begin{align*}
    &\left | \frac{\partial_{\tau} a}{a} \right | \sim \tau^{-1} \sim t^{p-1}\; , \qquad &&\left | \frac{a \partial_{\tau}^2 a - (\partial_{\tau} a)^2}{a^4} \right | \sim  \tau^{-\frac{2}{1-p}} \sim t^{-2} \; ,\\
    &\left | \frac{\partial_{\tau} a}{a^3} \right | \sim \tau^{1 - \frac{2}{1-p} } \sim t^{-1- p} \; ,  \qquad && \left | 1 + \tau \frac{\partial_{\tau} a}{a} \right | \sim 1 \; .
\end{align*}
\end{proposition}
\begin{proof}
We will repeatedly use \eqref{Laplace_Beltrami_FLRW} throughout the proof. 
We first analyse $[\square_{\mathbf{g}}, \partial_{\tau}]$. By definition, we have
\[
[\square_{\mathbf{g}}, \partial_{\tau}]\phi = -2\frac{\partial_{\tau} a}{a^3} \partial_{\tau}^2 \phi + 2 \left ( \frac{a \partial_{\tau}^2 a - 3 (\partial_{\tau} a )^2}{a^4} \right ) \partial_{\tau} \phi + 2 \frac{\partial_{\tau} a}{a^3} \Delta \phi.
\]
Expression \eqref{commutator_dtau} then follows by exploiting the definition of $\square_{\mathbf{g}} \phi$.
Next, notice that, by definition, we have: 
\[
[\square_{\mathbf{g}}, \Omega_{0j}] \phi = 2 \frac{\partial_{\tau} a}{a^3} x_j \partial_{\tau}^2 \phi + 2 \frac{\partial_{\tau} a}{a^3} \partial_j \phi - 2 x_j \frac{a \partial_{\tau}^2 a - 3 (\partial_{\tau} a)^2}{a^4} \partial_{\tau} \phi - 2 \frac{\partial_{\tau} a}{a^3} x_j \Delta \phi.
\]
Then, \eqref{commutator_boost} follows by applying the definition of $\square_{\mathbf{g}} \phi$ again.
Finally, for the scaling vector field:
\[
[\square_{\mathbf{g}} , S]\phi = -\frac{2}{a^2}\left ( 1 + \tau \frac{\partial_{\tau} a}{a}  \right ) \partial_{\tau}^2 \phi - \frac{2}{a^2} \left ( \frac{\partial_{\tau} a}{a} - \tau \frac{a \partial_{\tau}^2 a - 3 (\partial_{\tau} a)^2}{a^2} \right ) \partial_{\tau} \phi + \frac{2}{a^2} \left ( 1 + \tau \frac{\partial_{\tau} a}{a} \right ) \Delta \phi.
\]
Again, \eqref{commutator_scaling} follows from the expression of the Laplace-Beltrami operator and  from the fact that
\[
 \left( \frac{\partial_{\tau}a }{a^3} + \tau \frac{a \partial_{\tau}^2 a - (\partial_{\tau} a)^2}{a^4} \right) = 0,
\]
using \eqref{power_law}, \eqref{def_tau} and due to our definition of conformal time (see \eqref{conf_normalization}). 
\end{proof}
\begin{remark}
    The Friedmann equations for spatially-flat FLRW spacetimes read \cite{FriedmannAllDimensions}:
    \begin{align}
    \frac{\partial_{\tau}a}{a} &= \sqrt{\frac{8 \pi \rho_0}{3}}a^{1 - \frac{1}{p}}, \label{Friedmann1} \\
    \frac{a \partial_{\tau}^2 a - (\partial_{\tau} a)^2}{a^4}  &= -\frac{4 \pi}{3}(\rho + 3 P), \label{Friedmann2}
    \end{align}
    where $\rho$ and $P$ are, respectively, the energy density and pressure of the perfect fluid permeating the FLRW spacetime. Here, $\rho_0$ is the initial energy density. In particular, due to the equation of state $P = w \rho$ (with $w$ being the square of the speed of sound in the fluid),  notable cancellations for the above commutation relations occur in the case of vacuum spacetimes, for which $\rho \equiv 0$. The above expressions are related to the coefficients found in the commutation relations of proposition \ref{prop:commutation_relations}. In particular, the right hand sides of \eqref{Friedmann1} and \eqref{Friedmann2} are constant in the cases $p=0$ (Minkowski spacetime) and $p=1$ (Milne-like spacetime), as expected.
\end{remark}
\begin{remark}
    A naive approach to tackle the problem of global existence for quasi-linear wave equations in FLRW spacetime consists of following the known procedure for Minkowski space, where the generators of the Poincaré algebra and improved Sobolev identities are exploited, and trying to estimate the error terms stemming from the time-dependence of the FLRW metric. This approach, however, does not seem to work due to the  slow-decaying term
    \[
    \frac{\partial_{\tau}a}{a^3} \partial_j
    \]
    coming from commutation with the boost vector fields $\Omega_{0j}$, $j=1, 2, 3$ (see proposition \ref{prop:commutation_relations}).

    This obstruction motivates our choice to use Klainerman-Sideris types of estimates (see already lemma \ref{lemma:kl_sid_estimates}), so to commute only with the vector fields $\partial_{\mu}$, $S$ and $\Omega_{ij}$, with $\mu=0, 1, 2, 3$ and $i, j = 1, 2, 3$.
\end{remark}
\vspace{.2in}
We now introduce further notation to express the commutators of the Laplace-Beltrami operator with higher-order differential operators. Let $\Gamma \in \mathbb{P} = \{ \partial_{\mu}, S, \Omega_{ij}, \Omega_{0j} \}$. Then, it follows from proposition \ref{prop:commutation_relations} that
\[
[\square_{\mathbf{g}}, \Gamma] = f \square_{\mathbf{g}} + c^{\mu} \partial_{\mu},
\]
for some $f=f(\tau, x)$ and some vector $c^{\mu} = c^{\mu}(\tau, x)$. Table \ref{table:coefficients}  summarizes the results of proposition \ref{prop:commutation_relations} using this notation.
\begin{table}[h] 
\begin{center}
\renewcommand{\arraystretch}{1.8}
\begin{tabulary}{\textwidth}{c !{\vrule width 1pt} c l|c l|c l} 
\toprule  
& $f$ & & $c^0$ & &  $c^i$  & \\[.4em]
 \specialrule{1pt}{0pt}{0pt} 
 $\Gamma = \partial_i$ & 0 & &  0 & &  0  & \\  
 \rowcolor{gray!20}
 $\Gamma = \Omega_{ij} $ & 0 & &  0 & &  0 & \\ 
 $\Gamma = \partial_{\tau}$ & 2 $\frac{\partial_{\tau} a}{a}$ & $ = O(t^{p-1}) $ & $2\frac{a \partial_{\tau}^2 a - (\partial_{\tau} a)^2}{a^4}$ & $ = O(t^{-2})$ & 0 &  \\  
 \rowcolor{gray!20}
 $\Gamma = \Omega_{0j} $ & $-2x_j \frac{\partial_{\tau}a}{a} $ & $= x_j O(t^{p-1})$ & $-2 x_j \frac{a \partial_{\tau}^2 a - (\partial_{\tau} a)^2}{a^4}$ & $ = x_j O(t^{-2})$  & $2 \delta^{i}_j  \frac{\partial_{\tau}a}{a^3}  $ & $= O(t^{-1-p})$ \\ 
 $\Gamma = S$ & $2 \left ( 1 + \tau \frac{\partial_{\tau}a}{a} \right )$ & $ = O(1)$ & 0 &  & 0\\
\bottomrule
\end{tabulary}
\end{center}
\caption{Coefficients $f$ and $c^{\mu}$ related to the expansion $[\square_{\mathbf{g}}, \Gamma] = f \square_{\mathbf{g}} + c^{\mu} \partial_{\mu}$, where $\Gamma \in \{ \partial_{\mu}, S, \Omega_{ij}, \Omega_{0j} \}$.} \label{table:coefficients}
\end{table}

\begin{lemma}[Commutators of higher-order differential operators] \label{lemma:comm_n}
    Given $n \in \mathbb{N}$ and $\Gamma_1, $ $ \ldots, \Gamma_n $ $\in $ $ \mathbb{P} =$ $\{\partial_{\mu},$ $S, \Omega_{ij},$ $ \Omega_{0j}\}$, let:
    \[
    [\square_{\mathbf{g}}, \Gamma_i] = f_i \square_{\mathbf{g}} + c_i^{\mu}\partial_{\mu}, \quad i=1, \ldots, n \;,
    \]
    for some functions $f_i=f_i(\tau, x)$ and for some vectors\footnote{By a standard abuse of notation, we make no distinction between the vector $(c_i^{\mu})_{\mu = 0, \ldots, 3}$ and its components $c_i^{\mu}$, for $i$ fixed.} $c_i^{\mu} = c_i^{\mu}(\tau, x)$, $i = 1, \ldots, n$.
    
    Then, we have:
    \begin{align}
        [\square_{\mathbf{g}}, \Gamma_1 \cdots \Gamma_n] \phi  &= \Gamma_1 [\square_{\mathbf{g}}, \Gamma_2 \cdots \Gamma_n]\phi  \label{comm_higher_order} \\
        &+ c_1^{\mu} \partial_{\mu} \Gamma_2 \cdots \Gamma_n \phi + \label{firstline_lemma} \\
        &+f_1 c_2^{\mu} \partial_{\mu} \Gamma_3 \cdots \Gamma_n \phi + \\
        &+ f_1 (f_2 + \Gamma_2) c_3^{\mu} \partial_{\mu} \Gamma_4 \cdots \Gamma_n \phi + \ldots +\\
        &+ f_1 (f_2 + \Gamma_2) \cdot \ldots \cdot (f_{n-1} + \Gamma_{n-1}) c_n^{\mu}\partial_{\mu} \phi + \label{lastline_lemma} \\
        &+ f_1 (f_2 + \Gamma_2) \cdot \ldots \cdot (f_n + \Gamma_n) \square_{\mathbf{g}} \phi \;. \nonumber
    \end{align}
We stress that  lines \eqref{firstline_lemma}--\eqref{lastline_lemma} contain $n$ terms and we adopt the following convention: given $m \in \mathbb{N}$, $m < n$, then $[\square_{\mathbf{g}}, \Gamma_1 \cdots \Gamma_m]\phi$  is expressed as above, where only the first $m$ terms in \eqref{firstline_lemma}--\eqref{lastline_lemma} appear.
\end{lemma}
\begin{proof}
    Notice that, using  $[\square_{\mathbf{g}}, \Gamma_1] = f_1 \square_{\mathbf{g}} + c_1^{\mu}\partial_{\mu},$  we have:
    \begin{align*}
        \square_{\mathbf{g}} \Gamma_1 \cdots \Gamma_n \phi &= \Gamma_1 \square_{\mathbf{g}} \Gamma_2 \cdots \Gamma_n \phi + [\square_{\mathbf{g}}, \Gamma_1] \Gamma_2 \cdots \Gamma_n \phi \\
        & = (f_1 + \Gamma_1) \square_{\mathbf{g}} \Gamma_2 \cdots \Gamma_n \phi + c_1^{\mu}\partial_{\mu} \Gamma_2 \cdots \Gamma_n \phi \; .
    \end{align*}
    An iteration of the above argument yields:
    \begin{align}
        \square_{\mathbf{g}} \Gamma_1 \cdots \Gamma_n \phi &= (f_1 + \Gamma_1) \cdot \ldots \cdot (f_n + \Gamma_n) \square_{\mathbf{g}} \phi + \label{commutator_aux} \\
        &+  c_1^{\mu} \partial_{\mu} \Gamma_2 \cdots \Gamma_n \phi + (f_1 + \Gamma_1) c_2^{\mu}\partial_{\mu} \Gamma_3 \cdots \Gamma_n \phi + \ldots + \nonumber \\
        &+ (f_1 + \Gamma_1) \cdot \ldots \cdot (f_{n-1} + \Gamma_{n-1}) c_n^{\mu} \partial_{\mu} \phi  \; . \nonumber
    \end{align}
    On the other hand, using $[\square_{\mathbf{g}}, \Gamma_1] =  f_1 \square_{\mathbf{g}} + c_1^{\mu}\partial_{\mu} $ again:
    \begin{align}
        [\square_{\mathbf{g}}, \Gamma_1 \cdots \Gamma_n]\phi &=
        (\Gamma_1 \square_{\mathbf{g}} + [\square_{\mathbf{g}}, \Gamma_1])\Gamma_2 \cdots \Gamma_n \phi - \Gamma_1 \cdots \Gamma_n \square_{\mathbf{g}} \phi \nonumber \\ 
        &=\Gamma_1 [\square_{\mathbf{g}}, \Gamma_2 \cdots \Gamma_n] \phi +c_1^{\mu}\partial_{\mu}\Gamma_2 \cdots \Gamma_n \phi + f_1 \square_{\mathbf{g}} \Gamma_2 \cdots \Gamma_n \phi \; . \label{proof_commutation_rel}
    \end{align}
    The final expression then follows by using \eqref{commutator_aux} for the last term above.
\end{proof}
\begin{lemma}[Boost-free commutator estimates] \label{lemma:comm_estimate}
    Let $n, m \in \mathbb{N}$ with $n>0, m \ge 0$. Then, given $\Gamma_k \in \mathbb{L} = \{\partial_\mu, S, \Omega_{ij}\}$, $k = 1, \ldots, n+m$, we have:
    \[
    \left | \Gamma^{\le m}[\square_{\mathbf{g}}, \Gamma^{\le n}] \phi \right | \lesssim \left | \br{\square_{\mathbf{g}} \phi_{\le n-1}}_{\le m} \right | +  \left | \br{\square_{\mathbf{g}} \phi}_{\le m + n-1} \right | + t^{-2}  \left |\partial_{\tau} \phi_{\le m+n-1} \right | \; , 
    \]
    or, equivalently, in terms of $\partial_t$:
    \[
    \left | \Gamma^{\le m}[\square_{\mathbf{g}}, \Gamma^{\le n}] \phi \right | \lesssim \left | \br{\square_{\mathbf{g}} \phi_{\le n-1}}_{\le m} \right | + \left |\br{\square_{\mathbf{g}} \phi}_{\le m + n-1} \right |+  t^{-2+p}  \left | \partial_t \phi_{\le m + n-1} \right |  \; .
    \]
\end{lemma}
\begin{proof}
    First, we consider the case $m=0$. Moreover, let us momentarily assume, for the sake of concreteness, that $n \ge 5$. As it will be clear from the next steps, such an assumption can ultimately be dropped without any loss of generality.
    
    If we exploit \eqref{proof_commutation_rel} and \eqref{comm_higher_order}, we get:
    \begin{align}
        [\square_{\mathbf{g}}, \Gamma_1 \cdots \Gamma_n]\phi &= \Gamma_1 \Gamma_2[\square_{\mathbf{g}}, \Gamma_3 \cdots \Gamma_n] \phi + \Gamma_1 \br{ c_2^{\mu} \partial_{\mu} \Gamma_3 \cdots \Gamma_n \phi } + \label{first_row}\\
        &+\Gamma_1 \br{f_2 c_3^{\mu} \partial_{\mu} \Gamma_4 \cdots \Gamma_n \phi } + \nonumber \\
        &+ \Gamma_1 \br{f_2 (f_3 + \Gamma_3) c_4^{\mu} \partial_{\mu} \Gamma_5 \cdots \Gamma_n \phi } + \ldots + \nonumber \\
        &+ \Gamma_1 \br{ f_2 (f_3 + \Gamma_3) \cdot \ldots \cdot (f_{n-1} + \Gamma_{n-1} ) c_n^{\mu} \partial_{\mu} \phi } \nonumber \\
        &+ \Gamma_1 \br{ f_2 ( f_3 + \Gamma_3 ) \cdot \ldots \cdot (f_n + \Gamma_n) \square_{\mathbf{g}} \phi } \label{row_second_to_last} \\
        &+ \br{f_1 \square_{\mathbf{g}} + c_1^{\mu} \partial_{\mu} }\Gamma_2 \cdots \Gamma_n \phi \; . \label{last_row}
    \end{align}
    By inspecting table \ref{table:coefficients}, we observe that, for every $x \in \mathbb{R}^3$, the functions $\tau \mapsto f(\tau, x)$ and $\tau \mapsto c^{\mu}(\tau, x)$ associated to a vector field $\Gamma$ have  bounded derivatives of all orders, uniformly in $[\tau_0, +\infty)$, $\tau_0 > 0$. Moreover, $c^{\mu} = (c^0, 0)$,  since we are excluding the boost vector field, and $c^0$ decays in time. This decay is encoded in the following function:
    \begin{equation}
        \mathcal{C}(t) \coloneqq \sum_{\Gamma \in \{\partial_{\mu}, S, \Omega_{ij}\}}|c_{\Gamma}^0| \lesssim   \tau^{-\frac{2}{1-p}}  \sim  t^{-2}  \; ,
    \end{equation}
    see, for instance, table \ref{table:coefficients} again.
    Using the above, line \eqref{last_row} is bounded in absolute value by:
   \begin{align*}
    |f_1 \square_{\mathbf{g}} \Gamma_2 \cdots \Gamma_n \phi| + |c_1^{\mu} \partial_{\mu} \Gamma_2 \cdots \Gamma_n \phi| & =  
    |f_1 \square_{\mathbf{g}} \Gamma_2 \cdots \Gamma_n \phi| + |c_1^{0} \partial_{\tau} \Gamma_2 \cdots \Gamma_n \phi| \\
    &\lesssim |\square_{\mathbf{g}} \phi_{\le n-1}| + \mathcal{C}(t)|\partial_{\tau} \phi_{\le n-1}| \; .
     \end{align*}
    Similarly, line \eqref{row_second_to_last} is bounded by $|(\square_{\mathbf{g}} \phi)_{\le n-1}|$.

    On the other hand, all remaining terms (except for the first term at the right hand side of \eqref{first_row}) can be bounded, up to a multiplying constant, by 
    \[
    \mathcal{C}(t) \sum_{|\alpha| + |\beta| \le n-1} | \Gamma^{\alpha} \partial_{\tau} \Gamma^{\beta} \phi |  \lesssim \mathcal{C}(t) |\partial_{\tau} \phi_{\le n-1}| \; ,
    \]
    where the last inequality follows from commuting $\partial_{\tau}$ with the $\Gamma$ vector fields and using that $[\partial_{\tau}, \partial_{\mu}] = [\partial_{\tau}, \Omega_{ij}] = 0$ for every $\mu = 0, \ldots, 3$, for every $i, j = 1, 2, 3$ and $[\partial_{\tau}, S] = \partial_{\tau}$.

    Finally, the computations of lemma \ref{lemma:comm_n} can be applied iteratively to the first term at the right hand side of \eqref{first_row}: this gives additional terms that are akin to the terms we have already bounded, plus the term
    \[
    \Gamma_1 \cdots \Gamma_{n-1} [\square_{\mathbf{g}}, \Gamma_n] \phi = \Gamma_1 \cdots \Gamma_{n-1} (f_n \square_{\mathbf{g}} \phi + c_n^{0} \partial_{\tau} \phi) \; .
    \]
    Analogously to the previous computations, the above expression is bounded in absolute value by
    \[
    |(\square_{\mathbf{g}} \phi)_{\le n-1}| + \mathcal{C}(t) |\partial_{\tau} \phi_{\le n-1}| \; ,
    \]
    up to a multiplying constant.

    The proof in the case $m > 0$ is analogous, indeed the additional vector fields acting on $[\square_g, \Gamma_1 \cdots \Gamma_n] \phi$ (see \eqref{first_row}) can be controlled in a similar fashion, using that $[S, \partial_{\tau}] = \partial_{\tau}$ and that $[\Gamma, \partial_{\tau}] = 0$ if $\Gamma \in \mathbb{L} \setminus \{S\}$.
\end{proof}

\begin{corollary} \label{coroll:switch_brackets}
    Under the assumptions of lemma \ref{lemma:comm_estimate}, we have:
    \[
    \left | \br{\square_{\bf g} \phi_{\le n}}_{\le m} \right | \lesssim \left | \br{\square_{\bf g} \phi}_{\le m + n} \right | + t^{-2}  \left | \partial_{\tau} \phi_{\le m + n - 1} \right | \; ,
    \]
    or, equivalently, in terms of $\partial_t$:
    \[
    \left | \br{\square_{\bf g} \phi_{\le n}}_{\le m} \right | \lesssim \left | \br{\square_{\bf g} \phi}_{\le m + n} \right | + t^{-2 + p}  \left | \partial_t \phi_{\le m + n - 1} \right | \; .
    \]
\end{corollary}
\begin{proof}
    Using lemma \ref{lemma:comm_estimate}:
    \begin{align*}
        \left | \br{\square_{\bf g} \phi_{\le n}}_{\le m} \right | &\le \left | \Gamma^{\le m +n} \square_{\bf g} \phi \right | + \left | \Gamma^{\le m} [\square_{\bf g}, \Gamma^{\le n}] \phi \right | \\
        &\le \left |\br{\square_{\bf g} \phi}_{\le m+ n} \right | +
        \left | \br{\square_g \phi_{\le n-1}}_{\le m}\right |
        + t^{-2}\left | \partial_{\tau} \phi_{\le m + n -1 }\right | \, .
    \end{align*}
    By iterating the above argument for the second-to-last term above, the result follows.
\end{proof}

\begin{theorem}[Commuted energy estimate] \label{thm:commuted_energy}
    Let $\Gamma \in \mathbb{L} = \{ \partial_{\mu}, S, \Omega_{ij} \}$ and $0 < p < 1$. Then, using the notation in \eqref{def_len} and given $n \in \mathbb{N} \setminus \{ 0 \}$, we have:
    \begin{align}
    \sup_{t_0\leq t< t_1}E^{1/2}[\phi_{\le n}](t)+
    \lp{t^{p-1/2}\partial_t \phi_{\le n}}{L^2[t_0,t_1]}
    &\lesssim_{\,p, n}  E^{1/2}[ \phi_{\le n}](t_0)+\lp{t^{p+1/2}
   (\square_{\mathbf{g}}\phi)_{\le n}}{L^2[t_0,t_1]}\;,
    \label{energyEst_Gamma_n}
\end{align}
where $f\lesssim_{\,p} g$ means $f\leq C_p\,g$, with $C_p\rightarrow\infty$ as $p\rightarrow 0$. 
\end{theorem}
\begin{proof}
    We first consider the case $n = 1$. \newline
{\underline{Energy estimate for ${\partial_\tau\phi}$}}: Using the results of proposition \ref{prop:commutation_relations}, we derive the pointwise estimate
\begin{align*}
    t^{2p+1}|[\square_{\mathbf{g}},\partial_\tau]\phi|^2
    \lesssim 
    t^{2p+1} \left( t^{-4}|\partial_\tau\phi|^2 
    +t^{2p-2}|\square_{\mathbf{g}}\phi|^2\right)
    \lesssim t^{2p-2}\left(|t^{p-1/2}\partial_t\phi|^2 
    +|t^{p+1/2}\square_{\mathbf{g}}\phi|^2\right) \; .
\end{align*}
Applying the energy estimate \eqref{energyEst} to $\partial_\tau\phi$ (recall definition \eqref{energyDef})
and commuting with $\square_{\mathbf{g}}$ yields
\begin{align*}
    \sup_{t_0\leq t< t_1}E^{1/2}[\partial_\tau\phi](t)+
    \lp{t^{p-1/2}\partial_t\partial_\tau \phi}{L^2[t_0,t_1]}
    &\lesssim E^{1/2}[\partial_\tau\phi](t_0)+ \lp{t^{p+1/2}
    \partial_\tau (\square_{\mathbf{g}}\phi)}{L^2[t_0,t_1]} \\
    & +  \lp{O(t^{p-1}) t^{p-1/2}\partial_t \phi}{L^2[t_0,t_1]}
    +\lp{O(t^{p-1})t^{p+1/2}\square_{\mathbf{g}}\phi}{L^2[t_0,t_1]},
\end{align*}
where the underlying constant depends on $p$ as explained in the statement of the theorem; from now on we will keep on simplifying the notation by writing $\lesssim$ instead of $\lesssim_{\,p}$. 
Since $0 < p<1$, the $O(t^{p-1})$ quantity in the last two terms
is bounded (in fact decays in time).
Adding the energy estimate \eqref{energyEst} yields the following:
\begin{align*}
    \sup_{t_0\leq t< t_1}\sum_{k\leq 1}\left(E^{1/2}[\partial^{k}_\tau\phi](t)
    +\lp{t^{p-1/2}\partial_t\partial^{k}_\tau \phi}{L^2[t_0,t_1]}\right)
    &\lesssim \sum_{ k\leq 1}
    \left(E^{1/2}[\partial^{k}_\tau\phi](t_0)+ \lp{t^{p+1/2}
    \partial^{k}_\tau (\square_{\mathbf{g}}\phi)}{L^2[t_0,t_1]}
    \right) \; .
\end{align*}
\newline
{\underline{Energy estimate for $S\phi$}}: Applying the results of proposition \ref{prop:commutation_relations} once again, we get the pointwise estimate
\begin{align*}
    t^{2p+1}|[\square_{\mathbf{g}},S]\phi|^2
    \lesssim 
    t^{2p+1} |\square_{\mathbf{g}}\phi|^2 \; .
\end{align*}
Applying the energy estimate \eqref{energyEst} to $S\phi$
and commuting with $\square_{\mathbf{g}}$ yields
\begin{align}
    \sup_{t_0\leq t< t_1}E^{1/2}[S\phi](t)+
    \lp{t^{p-1/2}\partial_tS \phi}{L^2[t_0,t_1]}
    &\lesssim E^{1/2}[S\phi](t_0)+ \lp{t^{p+1/2}
    S (\square_{\mathbf{g}}\phi)}{L^2[t_0,t_1]} \label{energyEst_S}\\
    &\qquad 
    +\lp{ t^{p+1/2}\square_{\mathbf{g}}\phi}{L^2[t_0,t_1]}\notag  \; .
\end{align}
 After adding \eqref{energyEst}, we obtain:
\begin{align*}
    \sup_{t_0\leq t< t_1}\sum_{k\leq 1}\left(E^{1/2}[S^{k}\phi](t)
    +\lp{t^{p-1/2}\partial_tS^{k} \phi}{L^2[t_0,t_1]}\right)
    &\lesssim \sum_{k\leq 1}
    \left(E^{1/2}[S^{k}\phi](t_0)+ \lp{t^{p+1/2}
    S^{k} (\square_{\mathbf{g}}\phi)}{L^2[t_0,t_1]}
    \right) \; .
\end{align*}
\newline
{\underline{Energy estimate for $\phi_{\le 1}$}}:
Using the two energy estimates above, the improved energy bound
\eqref{energyEst} and the fact that \[ 
[\square_{\mathbf{g}}, \Omega_{ij}]=0=
 [\square_{\mathbf{g}}, \partial_i], \quad \forall\, i, j = 1, 2, 3 \; , 
 \]
we get:
\begin{align}
    \sum_{|\alpha|\leq 1}
    \Big(\sup_{t_0\leq t< t_1}E^{1/2}[\Gamma^{\alpha}\phi](t)+
    \lp{t^{p-1/2}\partial_t\Gamma^{\alpha}\phi}{L^2[t_0,t_1]}\Big)
    &\lesssim  \sum_{|\alpha|\leq 1}
    \Big(E^{1/2}[\Gamma^{\alpha} \phi](t_0)+\lp{t^{p+1/2}
    \Gamma^{\alpha}(\square_{\mathbf{g}}\phi)}{L^2[t_0,t_1]}
    \Big) \; .
    \label{energyEst_Gamma}
\end{align}
{\underline{Energy estimate for $\phi_{\le n}$}}: By induction, assume that \eqref{energyEst_Gamma} holds for $|\alpha| \le n-1$ (rather than $|\alpha| \le 1$).  Applying the energy estimate \eqref{energyEst} to  $\Gamma^{\le n} \phi$,  commuting  with $\square_{\bf{g}}$ and using the results of lemma \ref{lemma:comm_estimate}:
\begin{align}
    \sup_{t_0\leq t< t_1}E^{1/2}[\Gamma^{\le n}\phi](t)+
    \lp{t^{p-1/2}\partial_t \Gamma^{\le n} \phi}{L^2[t_0,t_1]} \label{energy_first_line}
    &\lesssim E^{1/2}[ \Gamma^{\le n} \phi](t_0)+ \lp{t^{p+1/2}
     \Gamma^{\le n} (\square_{\mathbf{g}}\phi)}{L^2[t_0,t_1]} \\
    &\qquad +\lp{t^{p+1/2} \square_{\bf{g}} \Gamma^{\le {n-1}} \phi}{L^2[t_0,t_1]}  \label{energy_first_term}\\
    &\qquad +\lp{t^{p+1/2}\Gamma^{\le {n-1}} (\square_{\mathbf{g}}\phi)}{L^2[t_0,t_1]}  \label{energy_second_term} \\
    &\qquad +\lp{t^{p-1/2} O \br{ t^{p-1}} \partial_t \Gamma^{\le {n-1}} \phi}{L^2[t_0,t_1]}. \label{energy_third_term}
\end{align}
Now, term \eqref{energy_second_term} can be immediately absorbed in the last term at the right hand side of \eqref{energy_first_line}. Term \eqref{energy_third_term} can be absorbed at the left hand side after taking a linear combination of the current inequality with the energy estimate for $\phi_{\le n-1}$ (which holds by induction assumption).  This procedure is analogous to that we followed for the $\partial_{\tau} \phi$ step.

Finally, the term \eqref{energy_first_term} can be absorbed by commuting $\square_{\bf{g}}$ with $\Gamma^{\le n-1}$. This yields additional terms akin to \eqref{energy_second_term} and \eqref{energy_third_term}, where $\Gamma^{\le n-2}$ replaces $\Gamma^{\le n-1}$. Repeating this procedure $n-1$ times is sufficient to conclude the proof.
\end{proof}

%
\section{Pointwise Decay Estimates}
\label{secPointwise} 
%
We first recall a lemma that we will exploit to obtain pointwise decay of solutions in the wave zone $\{r > \frac{\tau}{2}\}$.
\begin{lemma}
    Let $x \mapsto \phi(\tau,x)\in C^{\infty}_{c}(\mathbb{R}^{3})$, for each fixed value of $\tau$.
    Then, for every $\tau \ge 0$ and $r \ge 1$:
    \begin{equation}
        \label{ext_pointwise}
         |\phi|(\tau, r) \lesssim \frac{1}{r}\sum_{k =0}^{2}
        \sum_{|\beta|=0}^{2 }\lp{\partial^{k}_r
        \Omega^{\beta}\phi(\tau, \cdot)}{L^2} \ ,
    \end{equation}
    where $\Omega$ denotes the generators of spatial rotations.
\end{lemma}
\begin{proof}
    Since $\tau$ is fixed, the proof does not differ from that in Minkowski spacetime. In particular, the result follows from a scaling argument for the Sobolev inequality on homothetic domains away from the origin. We refer the reader to \cite{kl_remarks} for the details of the proof.
\end{proof}
We now present a series of lemmas that we will use later to control solutions in the interior zone $\{r \le \frac{\tau}{2} \}$; notice, however, that the result is more general both concerning the domain and the class of functions it applies to. 

\begin{lemma}(Klainerman-Sideris Type Estimates) \label{lemma:kl_sid_estimates}
    Suppose $\phi\in C^{\infty}([t_0,\infty)
    \times \mathbb{R}^3)$.
    Then:
    \begin{align}
        |(\tau-r)\Delta\phi|&\lesssim \sum_{|\alpha|\leq 1}
        \Big(|\nabla_{\tau,x}
        \Gamma^{\alpha}\phi|\Big)+
        a^2\tau|\square_{\mathbf{g}}\phi|\label{kl_sid1} \ , \\
         |(\tau-r)\partial^2_\tau\phi|&\lesssim \br{1 +  \frac{r}{\tau} }\sum_{|\alpha|\leq 1}
        \Big(
        |\nabla_{\tau,x}
        \Gamma^{\alpha}\phi|\Big)+
        a^2 r |\square_{\mathbf{g}}\phi|\label{kl_sid2} \ , \\
        |(\tau-r)\nabla_x \partial_\tau\phi|&\lesssim
        \sum_{|\alpha|\leq 1}
        \Big(|\nabla_{\tau,x}
        \Gamma^{\alpha}\phi|\Big)+
        a^2\tau|\square_{\mathbf{g}}\phi|\label{kl_sid3} \ , 
    \end{align}
where $\Gamma\in \mathbb{L} = \{ \partial_{\mu}, S, \Omega_{ij} \}$.
\end{lemma}
\begin{proof}
    Similarly to the proof of \cite[lemma 2.3]{kl_sid}, we have the identities
    \begin{align}
        \partial_\tau S\phi-\partial_\tau\phi &=
        \tau \partial^2_\tau\phi+r\partial_\tau\partial_r\phi \ , \label{identity1}
        \\  \partial_r S\phi-\partial_r\phi &=
        r \partial^2_r\phi+\tau \partial_\tau\partial_r\phi \ . \label{identity2}
    \end{align}
    Thus,
    \begin{equation}
    \label{scaling1}
        \tau(\partial_\tau S\phi-\partial_\tau\phi)
        -r(\partial_r S\phi-\partial_r\phi)=
        \tau^2\partial^2_\tau\phi-r^2\partial^2_r\phi \, .
    \end{equation}
    Using expression \eqref{Laplace_Beltrami_FLRW} for the Laplace-Beltrami operator in FLRW spacetimes:
    \begin{equation}
    \label{kl_sid_beg}
        \tau^2\partial^2_\tau\phi-r^2\partial^2_r\phi=
        (\tau^2- r^2)\Delta\phi +r^2(\Delta\phi -\partial_r^2\phi)
        -\tau^2a^2\square_{\mathbf{g}}\phi -2 \tau^2 \frac{\partial_\tau a}{a}
        \partial_\tau\phi \, .
    \end{equation}
    Rearranging, taking absolute values, and using
    the triangle inequality yields:
    \begin{equation}
    \label{kl_sid_mid}
        |\tau- r| |\Delta\phi|\leq \frac{1}{\tau+r}
        \left(|\tau^2\partial^2_\tau\phi-r^2\partial^2_r\phi|
        +r^2|\Delta\phi -\partial_r^2 \phi|
        +\tau^2a^2|\square_{\mathbf{g}}\phi| +2 \tau^2|\frac{\partial_\tau a}{a}|
        |\partial_\tau\phi|\right) \, .
    \end{equation}
By Lemma 2.1 in \cite{kl_sid} we have
\begin{equation}
\label{rotations1}
    \left|\Delta\phi -\partial_r^2\phi\right|
    \lesssim \frac{1}{r}\sum_{|\alpha|\leq 1}\sum_{1\leq j\leq 3}|
    \partial_j\Omega^{\alpha}\phi| \, .
\end{equation}
Going back to Eq. \eqref{kl_sid_mid}, we apply Eq. \eqref{rotations1}, Eq. \eqref{scaling1},
and the estimate $|\frac{\partial_\tau a}{a}|\sim \tau^{-1}$
to finish the proof of \eqref{kl_sid1}. Next,
returning to Eq. \eqref{kl_sid_beg} and rearranging terms we have
   \begin{equation}
        \tau^2\partial^2_\tau\phi-r^2\partial^2_r\phi=
        (\tau^2- r^2)\partial^2_\tau\phi +r^2(\Delta\phi -\partial_r^2 \phi)
        -r^2a^2\square_{\mathbf{g}}\phi -2 r^2\frac{\partial_\tau a}{a}
        \partial_\tau\phi \, .
    \end{equation}
Thus,
    \begin{equation}
        |\tau- r| |\partial^2_\tau\phi|\leq \frac{1}{\tau+r}
        \left(|\tau^2\partial^2_\tau\phi-r^2\partial^2_r\phi|
        +r^2|\Delta\phi -\partial_r^2 \phi|
        +r^2a^2|\square_{\mathbf{g}}\phi| +  2 r^2 |\frac{\partial_\tau a}{a}|
        |\partial_\tau\phi|\right) \, .
    \end{equation}
The rest of the proof of \eqref{kl_sid2} follows from the
same argument as before. Finally, \eqref{kl_sid3} also
follows from a similar argument after subtracting \eqref{identity2} to \eqref{identity1} and decomposing spatial
derivatives $\partial_x$ into radial and angular components.
\end{proof}
\begin{lemma}(Sobolev embedding in homogeneous spaces)
    Let $x \mapsto \phi(\tau,x)\in C^{\infty}_{c}(\mathbb{R}^{3})$ for each fixed value of $\tau$. Then:
    \begin{equation}
    \label{Sob_hom}
        \lp{\phi}{L^{\infty}}\lesssim 
        \lp{\nabla_{x}\phi}{L^{2}}+\lp{\Delta\phi}{L^2} \, .
    \end{equation}
\end{lemma}
\begin{proof}
    Given the Fourier transform of $\phi$, denoted by $\hat \phi$, we have:
    \[
        |\phi|(\tau, x) \lesssim \int_{\mathbb{R}^3} |\hat \phi|(\tau, \xi)d\xi  = \int_{B_1(0)} \frac{|\xi| | \hat \phi|(\tau, \xi)}{|\xi|} d\xi + \int_{\mathbb{R}^3 \setminus B_1(0)} \frac{|\xi|^2 | \hat \phi|(\tau, \xi)}{|\xi|^2}d\xi \; .
    \]
The proof then follows after applying the Cauchy-Schwarz inequality to each term at the right hand side, extending the integrals to $\mathbb{R}^3$ and applying Plancherel's theorem.
\end{proof}
\subsection{\texorpdfstring{Proof of Theorem \ref{global_ptwise_decay_linear}}{Proof of Theorem 2.3}}
\begin{proof}
Define smooth cutoff
functions
\begin{align*}
    \chiint &:= \begin{cases}
        1, & r \le \frac{\tau}{2}, \\
        h(\tau, r), & \frac{\tau}{2} \le r \le \frac34 \tau, \\
        0, & r > \frac34 \tau,
    \end{cases}
    \qquad
    \chiext := 1-\chiint\;,
\end{align*}
 where 
\[
h(\tau, r) \coloneqq g \br{ \frac{4r -2 \tau}{\tau} } 
\]
and $x \mapsto g(x)$ is a positive smooth function such that $g(x)=0$ for $x \ge 1$ and $g(x) = 1$ for $x \le 0$.
For a suitable choice of $h$,\footnote{A standard construction in terms of the function $x \mapsto e^{\frac{1}{x-1}}$ is sufficient for our purposes.} the cutoff functions satisfy the following properties:
\[
\chiint + \chiext = 1 \ ,
\qquad |\nabla^k_{\tau,x}(\chiint) | \lesssim_k \langle r \rangle ^{-k}
\chi_{ \{ \frac{\tau}{2} \le r \le \frac34 \tau \}}
\ , \qquad | \nabla^k_{\tau,x}(\chiext) | \lesssim_k \langle r \rangle ^{-k}
\chi_{ \{ \frac{\tau}{2} \le r \le \frac34 \tau \}} \ ,
\]
for $k \in \mathbb{N}$, $k \ge 1$, where $\chi_{D}$ is the indicator function of the set $D$ and $\langle \cdot \rangle = \sqrt{1 + (\cdot)^2}$ denote the Japanese brackets. In the following, we will employ the notation
\[
    \chi_{r \sim \frac{1}{2}\tau} := \chi_{ \left \{ \frac{\tau}{2} \le r \le \frac34 \tau \right \}}\;.
\]
We have:
\begin{align*}
    \lp{\tau\nabla_{\tau,x}\phi}{L^{\infty}}&=
    \lp{\tau(\chiint+\chiext)\nabla_{\tau,x}\phi}{L^{\infty}}\\
    &\lesssim
    \lp{\tau\nabla_{\tau,x}(\chiint\phi)}{L^{\infty}}
    +\lp{\tau\nabla_{\tau,x}(\chiext\phi)}{L^{\infty}}
    +\lp{\langle r \rangle ^{-1}\tau\chi_{r\sim \frac{1}{2}\tau}\phi}{L^{\infty}}\\
    &\lesssim \lp{\tau\nabla_{\tau,x}(\chiint\phi)}{L^{\infty}}
    +\lp{\tau\nabla_{\tau,x}(\chiext\phi)}{L^{\infty}}
    +\lp{ \phi }{L^{\infty}}\\
    &:=I+ I\!I+ I\!I\!I .
\end{align*}
{\underline{Bounding the term $I$:}}
Here we apply the Sobolev estimate \eqref{Sob_hom}:
\begin{align*}
    I&=\lp{\tau \nabla_{\tau,x}(\chiint\phi)}{L^{\infty}}\\
    & \lesssim \lp{\tau \nabla_{x}\nabla_{\tau,x} (\chiint\phi) (\tau, \cdot)}{L^{2}}
    +\lp{\tau \Delta \nabla_{\tau,x} (\chiint\phi) (\tau, \cdot)}{L^2}\\
    & \lesssim \lp{\tau \chiint\nabla_{x}(\nabla_{\tau,x}\phi)}{L^{2}}
    +\lp{\tau \chiint\Delta(\nabla_{\tau,x}\phi)}{L^2}
    +\lp{\tau \chi_{r\sim \frac{1}{2}\tau}
    (\langle r \rangle ^{-2}\phi, \langle r \rangle ^{-1}\nabla_{\tau,x}\phi,
    \langle r \rangle ^{-1} \nabla^2_{\tau,x}\phi)}{L^{2}},
\end{align*}
where, here and in the following, by $\nabla^2_{\tau, x}$ we mean $\nabla^{\alpha}_{\tau, x}$ with $|\alpha|=2$.
Next, we use the support properties of the cutoff:
For the terms containing $\chiint$, we use $\tau\lesssim |\tau-r|$. For the terms containing 
$\chi_{r\sim \frac{1}{2}\tau}$,
we use $\tau\lesssim r$ followed by
$|\chi_{r\sim \frac{1}{2}\tau }|\lesssim 1$. This yields:
\begin{align}
\lp{\tau \chiint\nabla_{x}(\nabla_{\tau,x}\phi)}{L^{2}}
    &+\lp{\tau \chiint\Delta(\nabla_{\tau,x}\phi)}{L^2}
     +\lp{\tau \chi_{r\sim \frac{1}{2}\tau}
    (\langle r \rangle ^{-2}\phi, \langle r \rangle ^{-1}\nabla_{\tau,x}\phi,
    \langle r \rangle ^{-1} \nabla^2_{\tau,x}\phi)}{L^{2}}\label{ptwise_inbetween}\\
    &\qquad \lesssim \sum_{2\leq |\alpha|\leq 3}
    \lp{\chiint (\tau-r)\nabla^{\alpha}_{\tau,x}\phi}{L^{2}}
    +\sum_{1 \le |\gamma|\leq 2}\lp{ \nabla_{\tau, x}^{\gamma} \phi  }{L^{2}} + \lp{ \frac{\phi}{ \langle r \rangle } }{L^2} \ .\notag 
\end{align}
All terms
containing at least
two spatial derivatives are controlled via the elliptic estimate
\[\lp{\nabla^{2}_{x}\phi}{L^{2}}\lesssim
\lp{\Delta \phi}{L^{2}} \ . \]
Therefore, we can apply the Klainerman-Sideris estimate
\eqref{kl_sid1} to the first term above  and commute 
to get
\begin{align}
    \sum_{2\leq |\alpha|\leq 3}
    \lp{\chiint (\tau-r)\nabla^{\alpha}_{\tau,x}\phi}{L^{2}}
    \lesssim 
    \sum_{1\leq |\alpha| \leq 2}
    \lp{\nabla^{\alpha}_{\tau,x}\phi_{\leq 1}}{L^2}+
    \sum_{|\beta| \leq 1}\lp{\tau a^2  \square_{\mathbf{g}}
    (\nabla^{\beta}_{\tau,x}\phi)}{L^2} \ . \label{ptwise_inbetween2}
\end{align}
Using the commutation relations \eqref{commutator_di} and \eqref{commutator_dtau}, and the fact that $\tau^{1 - \frac{2}{1-p}}a^2 \lesssim \tau^{-1}$,
we can simplify
the second term on the RHS above:
\begin{align}
    \lp{\tau a^2  \square_{\mathbf{g}}(\nabla^{\le 1}_{\tau,x}\phi)}{L^2}
    \lesssim \lp{a^2 \cdot (\tau+r) \nabla^{\le 1}_{\tau,x}(\square_{\mathbf{g}}\phi)}{L^2}
    +\lp{\tau^{-1} \nabla_{\tau,x}\phi}{L^2} \ . \label{ptwise_inbetween3}
\end{align}
Next, we go back to the last term on the RHS of Eq. \eqref{ptwise_inbetween}.
For the undifferentiated term in that sum, we apply
the following Hardy estimate \cite{Evans2}: 
\begin{equation}
    \lp{\langle r \rangle^{-1}\phi}{L^{2}} \lesssim \lp{r^{-1}\phi}{L^{2}}
    \lesssim \lp{\partial_r\phi}{L^{2}}\lesssim
    \lp{\nabla_{\tau,x}\phi}{L^2} \  .\label{hardy}
\end{equation}
Therefore:
    \[
    \sum_{1 \le |\gamma|\leq 2}\lp{ \nabla_{\tau, x}^{\gamma} \phi  }{L^{2}} + \lp{ \frac{\phi}{ \langle r \rangle } }{L^2}
    \lesssim \sum_{1\leq |\gamma| \leq 2}
    \lp{\nabla^{\gamma}_{\tau,x}\phi}{L^2} \  .\]
Combining these estimates along with \eqref{ptwise_inbetween2}
and \eqref{ptwise_inbetween3} yields 
\begin{align*}
    \sum_{2\leq |\alpha|\leq 3}
    \lp{\chiint (\tau-r)\nabla^{\alpha}_{\tau,x}\phi}{L^{2}}
    &+\sum_{1 \le |\gamma|\leq 2}\lp{ \nabla_{\tau, x}^{\gamma} \phi  }{L^{2}} + \lp{ \frac{\phi}{ \langle r \rangle } }{L^2}\\
     &\qquad \lesssim 
    \sum_{1\leq |\alpha| \leq 2}
    \lp{\nabla^{\alpha}_{\tau,x}\phi_{\leq 1}}{L^2}
    +\sum_{|\beta| \leq 1}\lp{a^2 \cdot (\tau+r) \nabla^\beta_{\tau,x}(\square_{\mathbf{g}}\phi)}{L^2} \, .
\end{align*}
By the energy estimate \eqref{energyEst_Gamma_n}
and the definition
of the energy norms, 
we get:
\begin{align*}
    I \lesssim &\sum_{1\leq |\alpha| \leq 2}
    \lp{\nabla^{\alpha}_{\tau,x}\phi_{\leq 1}}{L^2}
    +\sum_{|\beta| \leq 1}\lp{a^2 \cdot (\tau+r) \nabla^\beta_{\tau,x}(\square_{\mathbf{g}}\phi)}{L^2}\\
    &\qquad\qquad \lesssim   
    E^{1/2}
    [\phi_{\leq 2}](t)+
    \lp{t^{p-1/2}\partial_t
    \phi_{\leq 2}}{L^2[t_0,t_1]}
    +\lp{a^2 \cdot (\tau+r) (\square_{\mathbf{g}}\phi)_{\le 1} }{L^2}\\
    &\qquad\qquad \lesssim 
    E^{1/2}
    [\phi_{\leq 2}](t_0)
    +\lp{t^{p+1/2}
    (\square_{\mathbf{g}}\phi)_{\leq 2}}{L^2[t_0,t_1]}
    +\lp{a^2 \cdot (\tau+r)
    (\square_{\mathbf{g}}\phi)_{\le 1} }{L^2}
\end{align*}
{\underline{Bounding the term $I\!I$:}} Here we use 
bound \eqref{ext_pointwise}, the compact support of $\nabla\chiext$ and the fact that $[\Omega_{ij}, \partial_k] = \delta_{jk}\partial_i - \delta_{ik} \partial_j$ for every $i, j, k \in \{1, 2, 3\}$:
\begin{align*}
    I\! I&=\lp{\tau\nabla_{\tau,x}(\chiext\phi)}{L^{\infty}}\\
    &\lesssim \sum_{k =0}^{2}
    \sum_{|\beta|=0}^{2}\lp{\partial^{k}_r
    \Omega^{\beta}\nabla_{\tau,x}(\chiext\phi)(\tau, \cdot)}{L^2}\\
    & \lesssim 
       \sum_{k =0}^{2}\lp{\chiext
        \partial^{k}_{r}
        \nabla_{\tau,x}\phi_{\leq 2}}{L^2}
        +\lp{\chi_{r\sim \frac{1}{2}\tau}
        (\langle r \rangle ^{-1} \phi_{\le 2}, \langle r \rangle ^{-1}\nabla_{\tau,x}\phi_{\leq 2},
        \langle r \rangle^{-1}\nabla^2_{\tau,x}\phi_{\leq 2})}{L^{2}} \, .
\end{align*}
The last term is analogous to the last term of line \eqref{ptwise_inbetween}, which was already bounded.
In particular, for the undifferentiated terms, we again apply
the Hardy estimate \eqref{hardy}.
Then, by the energy estimate \eqref{energyEst_Gamma_n}: 
\begin{align*}
   I\! I \lesssim &\sum_{k =0}^{2}\lp{\chiext
        \partial^{k}_{r}
        \nabla_{\tau,x}\phi_{\leq 2}}{L^2}
        + \lp{\chi_{r\sim \frac{1}{2}\tau}
        (\nabla_{\tau,x}\phi_{\leq 2},
        \langle r \rangle^{-1}\nabla^2_{\tau,x}\phi_{\leq 2})}{L^{2}} \\
    &  \hspace{1.5in}
    \lesssim\sum_{k =0}^{2}\lp{
        \partial^{k}_{r}
        \nabla_{\tau,x}\phi_{\leq 2}}{L^2}
        +\lp{(\nabla_{\tau,x}\phi_{\leq 2},
        \langle r \rangle ^{-1}\nabla^2_{\tau,x}\phi_{\leq 2})}{L^{2}}\\
    &  \hspace{1.5in}\lesssim 
    E^{1/2}
    [\phi_{\leq 4}](t)
    +\lp{t^{p-1/2} \partial_t 
    \phi_{\leq 4}}{L^2[t_0,t_1]} \\
    &  \hspace{1.5in} \lesssim 
    E^{1/2}
    [\phi_{\leq 4}](t_0)
    +\lp{t^{p+1/2} 
    (\square_{\mathbf{g}}\phi)_{\leq 4}}{L^2[t_0,t_1]} \, .
 \end{align*}

{\underline{Bounding the term $I\!I\!I$:}} Here we apply
the Sobolev estimate \eqref{Sob_hom}
and exploit the same argument used for term $I$. The resulting
terms are harmless since they have equal or better decay in $\tau$, compared to I and II.

{\underline{Combining the estimates for $I$, $I \! I$ and $I \! I \! I$}:} We obtained:
\[
    \lp{\tau \nabla_{\tau, x} \phi}{L^{\infty}} \lesssim 
    E^{\frac12}[\phi_{\le 4}] (t_0) + 
    \lp{t^{p + \frac12}  \br{\square_{\bf g} \phi}_{\le 4}}{L^2[t_0, t_1]}
    + \lp{ a^2 \cdot (\tau + r)  \br{\square_{\bf g} \phi}_{\le 1} }{L^2} \, ,
\]
where we used, once again, definition \eqref{energyDef_tp_tau}. Now, applying the above inequality to $\phi_{\le k}$, for some $k \ge 0$, and using corollary \ref{coroll:switch_brackets} and the fact that $(\tau + r)a^2 \lesssim t^{1+p}$:
\begin{align*}
    \lp{\tau \nabla_{\tau, x} \phi_{\le k}}{L^{\infty}} &\lesssim \lp{\nabla_{\tau, x} \phi_{\le k+4} (t_0)}{L^2} + 
    \lp{t^{p + \frac12} \br{\square_{\bf g} \phi}_{\le k+4} }{L^2[t_0, t_1]} + \lp{t^{1+p} \br{\square_{\bf g} \phi}_{\le k + 1}}{L^2} \\
    &+ \lp{t^{-\frac12} O \br {t^{p-1} }  \partial_{\tau} \phi_{\le k + 3}}{L^2[t_0, t_1]} + 
    \lp{ O \br {t^{p-1} } \partial_{\tau} \phi_{\le k} }{L^2}.
\end{align*}
The last two terms can be reabsorbed in the remaining terms by means of \eqref{energyEst_Gamma_n}.  In fact, the last two terms are bounded by $E^{\frac12}[\phi_{\le k+3}](t) + \lp{t^{p-\frac12 } \partial_t \phi_{\le k+3}}{L^2[t_0, t_1]}$. 
This concludes the proof of \eqref{global_ptdecay_grad1}.
Estimate \eqref{global_ptdecay_grad2} follows immediately by exploiting
that $a(t)=t^{p}\sim \tau^{\frac{p}{1-p}}$.

\end{proof}

%
%
\section{\texorpdfstring{Proof of Theorem \ref{main1}, case $0 < p < 1$}{Proof of Theorem 2.4, case 0 < p < 1}}
\label{secGlobal} 


In this section, we establish small data global existence result
for the Cauchy problem~\eqref{FJequation}
under the conditions of smoothness and compact support for the initial data.
Throughout this section, $C_0$ will
denote the size of the initial data~\eqref{smallInitial_data_FJ}. We
start by defining
\begin{equation}
\label{bootDef}
\mathcal{E}_{K}(t):= 
\lp{t^p\nabla_{t,\tilde x}\phi_{\leq K}(t, \cdot )}{L^{2}}
 + C(p,K)\lp{t^{p-1/2}\partial_t\phi_{\leq K}}{L^2[t_0,t]}
\end{equation}
 for $t \ge t_0$ and for a suitable constant $C(p, K) > 0$ that can be derived from the proof of theorem \ref{thm:commuted_energy}.

We now proceed by a continuity argument. Let $K\geq 8$~.
We assume that $M$ is a large enough constant so that
\begin{equation}\label{initial_cond_boot}
	\mathcal{E}_K(t_0)\leq \frac{MC_0}{C_2} ,
\end{equation}
where $C_2$ is a constant that we will make explicit later.
Next, we assume as  bootstrap condition that  $T$ is the supremum over all $t\geq t_0$ for which  
\begin{equation}
\label{bootAss}
\sup_{t_0\leq t<T}\mathcal{E}_{K}(t) \leq 4MC_0\;.
\end{equation}
That a $T>t_0$ in such conditions exists is a direct consequence of
local existence of solutions, see e.g.\ \cite[section 4]{jaj}. We stress that the only constraint on $C_2$ is $C_2 > \frac14$.\\

As a result of the bootstrap assumptions,
the functions $\Gamma^{k}\phi$ enjoy the
pointwise decay summarized in the following:
\begin{lemma}
\label{ptwise_nonlinear}
    Let  $0\leq k \leq  K - 4$ and assume that \eqref{bootAss} holds. Then, for $C_0$
    small enough:
\begin{equation}
   \sup_{t_0\leq t<T} \lp{t\nabla_{t,\tilde x}\phi_{\leq k} (t, \cdot)}
   {L^{\infty}}\lesssim C_0 \ , \label{sup_boot1}
\end{equation}
where the underlying constants depend on $M$.
\end{lemma}
\begin{proof}
Let us first assume that $k \ge 4$ (we note that, in the present work, we take $K \ge 8$).
Using the bootstrap assumption and pointwise decay
estimate \eqref{global_ptdecay_grad2}
we see that for $0 \leq k\leq   K -4$, and $t_0\leq t<T$, we have 
 \[
      \lp{t \nabla_{t,\tilde x}\phi_{\leq k}
      (t,\cdot )}{L^{\infty}}
     \lesssim C_0
     +\lp{t^{p+1/2}
     (\partial_{t} \phi  \partial_{t} \phi)_{\leq k+4}}{L^2[t_0, T]}
     +\lp{t^{p+1}
     \big(\partial_{t} \phi  \partial_{t} \phi\big)_{\leq k+1}(t)}{L^2}
\]
For the first integral term at the right hand side, we distribute derivatives,  use that $[\partial_{\tau}, S ]=\partial_{\tau}$ and $[\Gamma, \partial_{\tau}] = 0$ if $\Gamma \in \mathbb{L} \setminus \{ S\}$,  the bootstrap assumption \eqref{bootAss}, the symmetry of the non-linearity and the fact that $\lfloor \frac{k+4}{2} \rfloor \le k$ when $k \ge 4$, to get:
    \begin{align}
    \lp{t^{p+1/2}
      \big(\partial_{t} \phi  \partial_{t} \phi\big)_{\leq k+4}}{L^2[t_0,T]} 
      &\lesssim \sum_{|\alpha_1| + |\alpha_2| = k+4} \lp{t^{p + \frac12} (\Gamma^{\alpha_1} \partial_t \phi ) (\Gamma^{\alpha_2} \partial_t \phi )}{L^2[t_0, T]} \nonumber \\
      &\lesssim \sup_{t_0\leq t<T}\lp{t \nabla_{t,\tilde x}\phi_{\leq \lfloor \frac{k+4}{2} \rfloor} (t, \cdot)}{L^{\infty}}
     \lp{t^{p-1/2}\partial_t\phi_{\leq k+4}}{L^2[t_0,T]} \notag \\
      &\lesssim \sup_{t_0\leq t<T}\lp{t \nabla_{t,\tilde x}\phi_{\leq k} (t, \cdot)}{L^{\infty}}
     \lp{t^{p-1/2}\partial_t\phi_{\leq k+4}}{L^2[t_0,T]}  \notag \\
     & \lesssim C_0 \cdot \sup_{t_0 < t < T} \lp{t \nabla_{t,\tilde x}\phi_{\leq k}(t, \cdot)}{L^{\infty}} \label{boot_rhs_1} \, .
\end{align}
We proceed similarly for the second term:
\begin{equation}
    \lp{t^{p+1}
      \big(\partial_{t} \phi  \partial_{t} \phi\big)_{\leq k+1}}{L^2}
     \lesssim \lp{t \nabla_{t,\tilde x}\phi_{\leq k}}{L^{\infty}}
     \lp{t^{p} \nabla_{t, \tilde x} \phi_{\leq k+1}}{L^2} 
     \lesssim \lp{t \nabla_{t,\tilde x}\phi_{\leq k}}{L^{\infty}}C_0 \label{boot_rhs_2} \, .
\end{equation}
Using equations \eqref{boot_rhs_1}--\eqref{boot_rhs_2}:
\[
\sup_{t_0 < t < T} \lp{t \nabla_{t,\tilde x}\phi_{\leq k}(t, \cdot)  }{L^{\infty}}\lesssim C_0 \left ( 1 
+ \sup_{t_0 < t < T}  \lp{t \nabla_{t,\tilde x}\phi_{\leq k} (t, \cdot) }{L^{\infty}} \right ) \, .
\]
The result follows by choosing $C_0$ small. Finally, since we proved the last inequality for $k \ge 4$, it must also hold for $k \ge 0$.
\end{proof}

Our next step in the main argument
is to apply the commuted energy estimate~\eqref{energyEst_Gamma_n}: 
 \begin{equation} 
 \label{energyEst2}
 \sup_{t_0\leq t<T}\mathcal{E}_{K}(t) 
 \lesssim \mathcal{E}_{K}(t_0)+\lp{t^{p+1/2}
     (\partial_{t} \phi  \partial_{t} \phi)_{\leq K}}{L^2[t_0,T]}
      :=\mathcal{E}_{K}(t_0)+I \ . 
\end{equation}
To bound the last term:
\[I\lesssim \lp{t^{p+1/2}(\partial_t\phi)_{\leq  k_1}
 (\partial_{t} \phi)_{\leq  k_2}}{L^2[t_0,T]}\]
with $ k=0,1,2,...,K$ and $k_1+ k_2= k$.
This implies that we have to estimate terms of the form 
\[     \Big(\int_{t_0}^{T}
 \!\!\int_{\mathbb{R}^3}t^{2p+1}\left(\Gamma^{ \alpha_1}
 \partial_{t} \phi  \Gamma^{ \alpha_2}\partial_{t} \phi  \right)^2 
 dxdt\Big)^{1/2} \ ,  \]
where $|\alpha_1| + |\alpha_2| = k$, $0 \le k \le K$.
By symmetry of the integrand
it suffices to treat the case $ |\alpha_1|\leq k/2$. Since $K\geq 8$
we have $0\leq  |\alpha_1| \leq k/2\leq K/2\leq K-4$ and therefore
we are allowed to use ~\eqref{sup_boot1}.
Thus, in view of the bootstrap assumption~\eqref{bootAss} and of the fact that $[\partial_{\tau}, S] = \partial_{\tau}$ and $[\partial_{\tau}, \Gamma] = 0$ if $\Gamma \in \mathbb{L} \setminus \{S\}$, we get: 
 \begin{align} 
 \label{Nterms2}
\Big(\int_{t_0}^{T}
 \!\!\int_{\mathbb{R}^3}t^{2p+1}\left(\Gamma^{ \alpha_1}
 \partial_{t} \phi  \Gamma^{ \alpha_2}\partial_{t} \phi  \right)^2 
 dxdt\Big)^{\frac{1}{2}}
&\lesssim
\Big(\int_{t_0}^{T}\lp{ t\Gamma^{ \alpha_1}
\partial_t\phi }{L^{\infty}}^2
 \!\!\int_{\mathbb{R}^3}t^{2p-1}\left(\Gamma^{ \alpha_2}
 \partial_{t} \phi  \right)^2 
 dxdt\Big)^{\frac{1}{2}}\\
 &\lesssim  \sup_{t_0\leq t<T} \lp{t\nabla_{t,\tilde x}
 \phi_{\leq K-4}}{L^{\infty}}\Big(\int_{t_0}^{T}
 \!\!\int_{\mathbb{R}^3} t^{2p-1}
 \left(\Gamma^{\alpha_2}\partial_{t} \phi  \right)^2 
 dx dt\Big)^{\frac{1}{2}} \notag \\
 &\lesssim C_0\Big(\int_{t_0}^{T}
 \!\!\int_{\mathbb{R}^3} t^{2p-1}
 \left(\Gamma^{\alpha_2}\partial_{t} \phi  \right)^2 
 dx dt\Big)^{\frac{1}{2}} \notag \\
 &\lesssim C_0\lp{t^{p-1/2}\partial_t\phi_{\leq K}}{L^2[t_0, T]}
\lesssim C_0   \sup_{t_0\leq t<T} \mathcal{E}_K(t)\;. \notag 
\end{align}
To conclude, we can sum these estimates to show that
the bulk term satisfies the bilinear bounds 
 \begin{equation} 
 \label{nonlinEst}
 I\lesssim \lp{t^{p+1/2}(\partial_t\phi)_{\leq  k_1}
 (\partial_{t} \phi)_{\leq  k_2}}{L^2[t_0,T]}
 \lesssim 
 C_0   \sup_{t_0\leq t<T} \mathcal{E}_K(t)\;.
 \end{equation}
We are now ready to close our bootstrap argument.
From estimates ~\eqref{energyEst2} and
~\eqref{nonlinEst}, it follows that
\begin{equation} 
\label{energyEst4}
 \sup_{t_0\leq t<T} \mathcal{E}_{K}(t)  
\leq C_2 \left(\mathcal{E}_{K}(t_0)
 +C_0 \sup_{t_0\leq t<T} \mathcal{E}_{K}(t) \right), 
  \end{equation}
where $C_2$ is the maximum of the constants in estimates
~\eqref{energyEst2} and ~\eqref{nonlinEst}.
Choosing $C_0$ sufficiently small we can ensure that
\[C_0C_2< \frac{1}{2} .\]
Rearranging \eqref{energyEst4}:
\[
 \sup_{t_0\leq t<T} \mathcal{E}_{K}(t)  
\leq 2C_2  \mathcal{E}_{K}(t_0),  
\]
Applying \eqref{initial_cond_boot}:
 \begin{equation*} 
  \sup_{t_0\leq t<T}\mathcal{E}_{K}(t) \leq 2MC_0 ,
  \end{equation*}
which corresponds to a strict improvement of the bootstrap
assumption~\eqref{bootAss}, from which we can conclude that $T=+\infty$.

\section{\texorpdfstring{Proof of Theorem \ref{main1}, case $p=1$}{Proof of Theorem 2.4, case p=1}} \label{secp1}
Here we follow the proof of section \ref{secGlobal} to show small data global existence in the case $p=1$. Differently from the previous case, no commutation with the vector fields of $\mathbb{L}$ is needed.

We define  
\begin{equation}
\label{bootDefp1}
\mathcal{E}_{K}(t):= 
\lp{t \nabla_{t,\tilde x} \partial_x^{\le K} \phi(t, \cdot )}{L^{2}}
 + C(K)\lp{t^{1/2}\partial_x^{\le K} \partial_t\phi}{L^2[t_0,t]}
\end{equation}
for $t \ge t_0$. Let $M, C_0, C_2 > 0$ be such that
\[
\mathcal{E}_{K}(t_0) \le \frac{M C_0}{C_2} \, .
\]
As a bootstrap assumption, we assume that $T$ is the supremum over all $t \ge t_0$ for which:
\begin{equation}
    \sup_{t_0 \le t < T} \mathcal{E}_K(t) \le 4 M C_0 \, .
\end{equation}
By the energy estimate \eqref{energyEst}:
\[
\mathcal{E}_K(t) \lesssim \mathcal{E}_K(t_0) + \lp{ t^{3/2} \partial_x^{\le K} \br{\partial_t \phi \partial_t \phi} }{L^2[t_0, t]} = \mathcal{E}_K(t_0) + I \, ,
\]
for every $t_0 \le t < T$.
We notice that  I is composed by terms of the form
\[
\lp{ t^{3/2} \partial_x^{k_1} \partial_t \phi  \partial_x^{k_2} \partial_t \phi }{L^2[t_0, t]},
\]
where $k_1 + k_2 \le K$ and $k_1, k_2 \ge 0$. Without loss of generality, we assume that 
$0 \le k_1 \le K/2 < K-3/2$. By Sobolev embedding, for $0 \le k_1 < K- \frac32$:
\[
\lp{ \partial_x^{k_1} \partial_t \phi (t, \cdot) }{L^{\infty}} \le C \lp{ \partial_t \phi(t, \cdot) }{H^K} \lesssim \frac{\mathcal{E}_K(t)}{t} \le \frac{4 M C_0}{t},
\]
where in the last step we exploited the bootstrap assumption. Hence:
\[
I \lesssim 4M C_0 \lp{ t^{\frac12} \partial_x^{\le K} \partial_t \phi  }{L^2[t_0, t]} \lesssim 4MC_0 \mathcal{E}_K(t).
\]
Then, similarly to the procedure of section \ref{secGlobal}, the bootstrap closes provided that $C_0$ is chosen sufficiently small.

\section*{AI Statement}
The main mathematical results of the paper, which remain unchanged, were first posted to arXiv in March 2026 and were obtained without AI assistance. AI tools subsequently helped check the scaling commutator for power-law FLRW metrics, leading to minor proof simplifications, and identify references on wave equations with time-dependent damping in three spatial dimensions. The authors take responsibility for the final manuscript.

\section*{Acknowledgements}
JLC was partially supported by FCT/Portugal through CAMGSD, IST-ID (project UID/04459/2025). JO was partially supported by an AMS-Simons Research
Enhancement Grant for Primarily Undergraduate Institution (PUI) Faculty. JO gratefully acknowledges the hospitality of the Instituto Superior Técnico (IST), Lisbon, during his stay, when part of this work was completed. FR was supported by a postdoctoral research fellowship from the Gran Sasso Science Institute and, during the first part of the project, was supported by FCT/Portugal through the PhD scholarship UI/BD/152068/2021.

\end{document}